\documentclass[10pt,conference,letterpaper]{IEEEtran}

\usepackage{array}
\usepackage{cite}

\usepackage{mdwmath}
\usepackage{mdwtab}

\usepackage{graphicx}
\usepackage{float}
\usepackage{stfloats}
\usepackage{lipsum}
\usepackage{tabularx}
\usepackage[tight,footnotesize]{subfigure}
\usepackage{caption}

\usepackage{amsmath}
\usepackage{amssymb}
\usepackage{paralist}
\usepackage{mathrsfs}
\usepackage{url}
\usepackage{mathbbol}
\usepackage{times}
\usepackage{booktabs}
\usepackage{verbatim}

\usepackage{bm}			

\usepackage{amsthm}

\newtheorem{definition}{Definition}
\newtheorem{problem}{Problem}
\newtheorem{lemma}{Lemma}

\newtheorem{corollary}{Corollary}
\newtheorem{assumption}{Assumption}
\newtheorem{result}{Result}

\usepackage{color}

\newcommand{\REMOVED}[1]{}

\newcommand{\eq}[1]{Eq.~\eqref{#1}}
\newcommand{\eqs}[1]{Eqs.~\eqref{#1}}

\usepackage[normalem]{ulem}
\newcommand{\redout}{\bgroup\markoverwith{\textcolor{red}{\rule[.3ex]{2pt}{1pt}}}\ULon}

\makeatletter
\newcommand{\pushright}[1]{\ifmeasuring@#1\else\omit\hfill$\displaystyle#1$\fi\ignorespaces}
\newcommand{\pushleft}[1]{\ifmeasuring@#1\else\omit$\displaystyle#1$\hfill\fi\ignorespaces}
\makeatother

\newcommand\numberthis{\addtocounter{equation}{1}\tag{\theequation}}

\begin{document}
\title{Offloading on the Edge: Analysis and Optimization of Local Data Storage and Offloading in HetNets}
 \author{
 \IEEEauthorblockN{Pavlos Sermpezis}
 \IEEEauthorblockA{Mobile Communications 
 Dept.\\EURECOM, France\\pavlos.sermpezis@eurecom.fr}
\and
 \IEEEauthorblockN{Luigi Vigneri}
 \IEEEauthorblockA{Mobile Communications Dept.\\EURECOM, France\\luigi.vigneri@eurecom.fr}
 \and
 \IEEEauthorblockN{Thrasyvoulos Spyropoulos}
 \IEEEauthorblockA{Mobile Communications Dept.\\EURECOM, France\\thrasyvoulos.spyropoulos@eurecom.fr}
 }
\maketitle

\begin{abstract}
The rapid increase in data traffic demand has overloaded existing cellular networks. Planned upgrades in the communication architecture (e.g. LTE), while helpful, are not expected to suffice to keep up with demand. As a result, extensive densification through small cells, caching content closer to or even at the device, and device-to-device (D2D) communications are seen as necessary components for future heterogeneous cellular networks to withstand the data crunch. Nevertheless, these options imply new CAPEX and OPEX costs, extensive backhaul support, contract plan incentives for D2D, and a number of interesting tradeoffs arise for the operator. 
In this paper, we propose an analytical model to explore how much local storage and communication through ``edge'' nodes could help offload traffic in various heterogeneous network (HetNet) setups and levels of user tolerance to delays. We then use this model to optimize the storage allocation and access mode of different contents as a tradeoff between user satisfaction and cost to the operator. Finally, we validate our findings through realistic simulations and show that considerable amounts of traffic can be offloaded even under moderate densification levels.
\end{abstract}

\section{Introduction}
The growth in the number of ``smart'' mobile devices and connection speeds has led to a high volume of mobile data traffic. Cellular networks are currently overloaded and, despite a lot of planned improvements on the physical layer technologies, they are not expected to be able to keep up with the rapidly increasing user data demand~\cite{cisco2014}. Radically reducing the communication distance by deploying, and \emph{offloading} traffic to, many ``small cells'' (e.g. femto, pico, or even WiFi) is seen as the only viable solution~\cite{femtocell-survey,HetNets-paradigm,Offloading-Wifi}. Nevertheless, this requires a large investment in upgrading the backhaul network, increasingly based on wireless links, which will often be the new performance bottleneck~\cite{femtocaching}. Caching popular content at the ``edge'', i.e. on storage devices installed at small cell base stations could alleviate backhaul congestion~\cite{femtocaching,poularakis-video-HetNets}, and is supported by a number of real data studies suggesting a high amount of demand overlap between user requests~\cite{youtube-traffic-from-edge,top-video-cellular,pptv-mobile-vod}.

Reducing the communication distance is taken yet a step further with the newly proposed paradigm of device-to-device (D2D) communication~\cite{survey-d2d,caire-d2d-caching-limits}. A device can store a (popular) content after consuming it, and give it directly to other neighboring devices also interested in it, offloading these requests from the main network. The connection between the two devices could be in-band (cellular frequencies) or out of band (e.g. Bluetooth, WiFi Direct). While D2D-based offloading normally assumes a content request will either be served \emph{immediately} from a device currently in range or the cellular network, some recent works have suggested the use of \emph{opportunistic offloading} through D2D: a device requesting some content might wait for some amount of time until it \emph{encounters} another device sharing the content~\cite{multiple-offloading,fluid-limit-mass2012,offloading-control-theory}, and go back to the main network if not found before some set deadline. 

Hence, more data could be offloaded from the main network through such D2D communication, perhaps at the expense of increased delay for some requests. Such increased delays could sometimes be acceptable (e.g. asynchronous requests, longer start-up or buffering delays easily amortized when considering large content). Yet, in many cases, the operator will need to provide appropriate incentives to these users, either in the form of instantaneous price reductions~\cite{tube} or low(er) priced plans. 
What is more, operators will probably need to also provide incentives to the devices storing the content and acting as local \emph{relays} on their behalf, as this raises important battery consumption, storage, as well as privacy and security issues. 

The provision of these incentives constitutes another important form of cost for the operator, together with the costs of directly serving the content from the main (mostly macro-cell based) network, and that of installing, maintaining, and supporting with ample backhaul capacity, new small cells with large enough caches. It thus becomes increasingly important for an operator of such a future Heterogeneous Network (HetNet) with caching and D2D capabilities to be able to answer questions like: \emph{"How much content can be offloaded by a given setup as a function of content demand patterns?"}, \emph{"Is it worth investing in additional cell densification, or would it be more cost-efficient to provide incentives for D2D opportunistic offloading?"}.   

To this end, in this paper we propose an analytical model that can be used to study the problem of "offloading on the edge" in a HetNet. Although capturing all the fine details of possible setups and technologies would be a rather daunting task, we assume two main mechanisms being employed in the considered network, namely (i) caching on small cells and mobile devices, collectively referred to as "edge nodes", and (ii) offloading requests through local, short range communications (e.g. D2D or low power communication to local femto or pico base stations).  We describe the "offloading on the edge" mechanism and propose a generic model that allows us to analytically study it (Section~\ref{sec:system-model}). We proceed by deriving useful results for the performance of content delivery through this mechanism and the incurred costs, as a function of key system parameters (Section~\ref{sec:analysis-single-content}). Then, we study the total offloading cost and provide insights for content placement and dissemination strategies that minimize this cost (Section~\ref{sec:single-cost-optimization}). Finally, we validate our results through realistic simulations (Section~\ref{sec:single-validation}) and discuss related work (Section~\ref{sec:related}).

Summarizing, the main contributions of our work are:
\begin{itemize}
\item To our best knowledge, this is the first work jointly and analytically studying offloading through small cells, opportunistic D2D, and caching at both.   
\item We provide closed-form analytical approximations applicable to a number of performance metrics and network setups.
\item We provide initial insights into the various design tradeoffs involved, as well as the efficient allocation of storage space among different contents. 
\end{itemize}

\section{Offloading on the Edge}\label{sec:system-model}
\subsection{Network Setup}
We consider a Heterogeneous Cellular Network (HetNet)~\cite{HetNets-paradigm}, composed of 3 sets of nodes: 

\textit{Macro-cell Base Stations} ($\mathcal{BS}$): They provide full coverage to subscribed mobile nodes (MNs), but we assume their radio resources are congested. 

\textit{Small Cells} ($\mathcal{SC}$): These are shorter range, low power base stations (e.g. femto and pico-cells, or even WiFi access points) dispersed in the area of coverage. They provide ample capacity to the few MNs within range, and their communication cost to/from a MN is smaller~\cite{johansson2007cost}. Hence, they can be used to offload some traffic from BSs. However, the backhaul connection for these cells will often be wireless (either to a BS or to an aggregation point) and underprovisioned~\cite{femtocaching}. This makes a backhaul transmission to a small cell costly. To this end, each small cell is equipped with some storage capacity, as in~\cite{femtocaching,poularakis-video-HetNets}, where (popular) content could be cached to avoid duplicate backhaul accesses.

\textit{Mobile Nodes} ($\mathcal{MN}$): These include smartphones, tablets, netbooks, etc. MNs can communicate with BSs, SCs (if in range), and even other MNs directly, if D2D communication is allowed. D2D communication potentially offers higher rates at lower interference levels~\cite{survey-d2d}. Yet, appropriate incentives from the operator might be needed.  Without loss of generality, we assume out-of-band communication (e.g. WiFi Direct or Bluetooth) for D2D. We also assume that each MN also has some storage capacity (normally less than that of a small cell) for caching (popular) content. 

The number of nodes in each set is
\begin{equation*}
N_{BS} = |\mathcal{BS}|~~,~~N_{SC} = |\mathcal{SC}|~~,~~N_{MN} = |\mathcal{MN}|
\end{equation*}
where $|\cdot|$ denotes the cardinality of a set.

\subsection{Offloading Mechanism}\label{sec:system-communication}

\textbf{\textit{Content Requests.}} We assume that each MN is interested in different contents over time (e.g. videos, web pages, software updates, etc.), and that the same content may be of interest to multiple MNs. This interest overlap is supported by recent studies (e.g.~\cite{youtube-traffic-from-edge,top-video-cellular,pptv-mobile-vod}, to name a few), where the popularity distribution of contents is shown to be highly skewed. In the remainder, we will be assuming that the number of nodes interested in a content, the content popularity, is known in advance or can be estimated. For a number of applications, like \textit{push services}~\cite{fluid-limit-mass2012}
, this information can be known in advance by the cellular network. Users are subscribed to a push service they are interested in (e.g. news, series episodes, trending videos, etc.), and when a content (of this service) is created or published, the content provider starts distributing (\textit{pushing}) it to them\footnote{We assume that the content provider may be the cellular network operator itself or in cooperation with it (like the Akamai and Swisscom example~\cite{akamai-swisscom}).}.  Similarly, users might subscribe to certain categories of contents, such as personalized Internet radio stations like Pandora and Jango\footnote{\url{www.pandora.com} , \url{www.jango.com}}. The content of these pseudo-random streams of songs can be decided in advance, and thus the popularity of songs belonging to different streams can be estimated. 

\textbf{\textit{Content Delivery.}} An operator can deliver a content to an interested MNs in one of the following ways: (i) \textit{Direct transmission} from a BS; (ii) \textit{Offloading through SCs and/or MNs}, where the operator transmits the content to some SCs over the backhaul and stores it there, or instructs some MNs to store a content for some time (e.g. keeping in their cache a content they consumed). Then, the operator can ask an interested MN within range of a SC or MN caching that content to retrieve it directly.

Moreover, an operator can ask an MN interested in a content $\theta$, but not \emph{currently} within range of an SC or MN with content $\theta$ in its cache, to wait for an amount of time, let $TTL$, until it \emph{moves} within range of such an SC or MN. If this time expires, then the operator is obliged to deliver the content directly through the closest macro BS. While this \textit{delay-tolerant} approach is in contrast to the usual ones considered for small cell and D2D based offloading~\cite{femtocaching,caire-d2d-caching-limits,poularakis-video-HetNets}, it is likely that the small cell and (D2D enabled) mobile node density will not always be enough to offload enough traffic. Hence, it is a valuable (and complementary) alternative, with potential benefits (increased offloading) and costs (reduced QoE and potential monetary incentives)\footnote{Clearly, such delays might not be acceptable for all applications. However, many applications are inherently delay-tolerant, e.g. software updates, file downloads, one way streaming (e.g. YouTube or Netflix). Moreover, users might be willing to accept small or larger delays, if appropriate incentives are provided, and delayed content delivery has already been considered in a number of contexts, e.g~\cite{tube,push-to-peer}
.}.

\subsection{Cost Model}\label{sec:cost-model}
The goal of an offloading mechanism is to minimize the cost of delivering a set of contents over time to different nodes. Hence, we need first to define a model for the costs involved in each phase of the "offloading on the edge" mechanism.\\
\noindent \textbf{\textit{$-$ Initial Placement Costs: $C_{BH}$, $C_{BS}$.}}\\
The content provider, at time $t_{0}=0$, places the content to some edge nodes (SCs and/or MNs). A content is placed to a SC through a backhaul (wired or wireless) transmission, and we denote this per placement cost as $C_{BH}$. A (possible) content placement to some MNs takes place through a macro-cell BS transmission. We denote this transmission cost, which mainly depends on the load/congestion of the BSs, as $C_{BS}$.

\noindent \textbf{\textit{$-$ Opportunistic Offloading Costs:  $C_{SC}$, $C_{D2D}$.}}\\
During time $t\in(0,TTL]$, the holders (which are either SCs or MNs) deliver the content to any requester they meet. We consider different costs for a SC-MN and a MN-MN (or D2D) transmission: $C_{SC}$ and $C_{D2D}$. The former cost depends on the operating cost (transmission, energy consumption) of an SC, whereas the latter might exist if a compensation (or reward) is given by operator to MNs for each content they offload.

\noindent \textbf{\textit{$-$ Delayed Delivery Cost: $ C_{BS}^{(TTL)}$.}}\\
At time $TTL$, the cellular network sends through macro-cell BSs the content to every non-served requester. This cost relates both to the load of BS (as $C_{BS}$) and to a (possible) compensation to the MNs for a delayed delivery. We denote this (per transmission) cost as $C_{BS}^{(TTL)}$.

\subsection{Content Dissemination Model and Assumptions}\label{sec:analysis-preliminaries}
Let us assume a content item (e.g. a popular video file) and a set of MNs interested in it. The content provider, at time $t_{0}=0$, places the content to the caches of some SCs and/or MNs. If by an expiry time $TTL$ (if any), some of the interested MNs have not met any edge node (SC or MN) with the content, they are served by a macro-cell BS\footnote{In the mechanism we consider, the content is cached only at the initial time, $t_{0}=0$, and macro-cell BSs deliver it only at its expiry time , $t=TTL$. Although one could place contents during time $t\in(0,TTL)$ as well, it has been shown (for similar settings) that placing contents at times $t\in(0,TTL)$ leads to a sub-optimal performance~\cite{offloading-control-theory,fluid-limit-mass2012}.}.

For the ease of reference, we define the following sets of "edge nodes" that are involved in the offloading process: 
\begin{definition}
A \emph{requester} of a content is a mobile node (MN) that (a) is interested in the content and (b) has not received it yet. We denote the set of requesters at time $t$ as $\mathcal{R}(t)$.
\end{definition}
\begin{definition}
A \emph{holder} of a content is an edge node (SC or MN) that stores the content and will forward it to its requesters. We denote the set of holders at time $t$ as $\mathcal{H}(t)$.
\end{definition}
\noindent We further denote the number of requesters and holders as:
\begin{equation*}
R(t) = |\mathcal{R}(t)|~~\text{and}~~H(t) = |\mathcal{H}(t)|
\end{equation*}
where $\mathcal{H}(t) = \mathcal{H}_{SC}(t)\cup\mathcal{H}_{MN}(t)$ and $H(t) = H_{SC}(t)+H_{MN}(t)$

To model the level of participation of MNs in the offloading mechanism, we make the following assumption.
\begin{assumption}[Cooperation]\label{ass:cooperation} A requester acts as a holder for a content it has received with probability $p_{c}\in[0,1]$. The probability $p_{c}$ is equal among all nodes and contents.
\end{assumption}

The probability $p_{c}$ captures either the chance a node to forward the content (e.g. it has enough resources at the time) or the percentage of nodes who are "contracted" to help\footnote{ Here, we need to stress that the above assumption implies that MNs will never become holders of a content they are not interested in. Although there exist studies that assume that even not interested MNs might be willing to act as holders~\cite{offloading-wowmom11,offloading-control-theory,fluid-limit-mass2012,offloading-double-opportunities}, we believe that incentive mechanisms for these cases are difficult to implement (e.g. a user easier accepts to forward a content it already has stored, than to retrieve, cache and forward a content it will never use). Nevertheless, our framework could be easily extended also for such cases.}.

Finally, since edge nodes can exchange data only when they come within transmission range, the offloading is heavily affected by these \textit{meeting events} between nodes. We assume the following class of node mobility.
\begin{assumption}[Mobility]\label{ass:heterogeneous-mobility}~\\
$-$ The meeting events between two nodes $\{i,j\}$, $i\in\mathcal{MN}$ and $j\in\mathcal{MN}\cup\mathcal{SC}$, are given by a Poisson process with rate $\lambda_{ij}$.\\
$-$ The meeting rates $\lambda_{ij}$ are drawn from an (arbitrary) probability distribution $f_{\lambda}(\lambda)$ with mean value $\mu_{\lambda}$.\\
$-$ Meeting duration is negligible compared to the time intervals between nodes, but long enough for a content exchange.
\end{assumption}
Assumption~\ref{ass:heterogeneous-mobility} is a tradeoff between realism (heterogeneous $\lambda_{ij}$) and tractability (Poisson process). Heterogeneous meeting rates are motivated by analysis of real mobility traces~\cite{Gao2009,Conan2007}, where not all people meet each other with the same frequency, and by the different communication ranges (SC-MN and MN-MN). Similar assumptions are common in related works~\cite{Gao-user-centric-DTN,multiple-offloading,offloading-double-opportunities,fluid-limit-mass2012,offloading-control-theory} and have been shown to not be far from real mobility~\cite{Gao2009,Conan2007}. Yet, in Section~\ref{sec:single-validation}, we test our results against realistic scenarios where node mobility departs from our assumptions and involves much more complexity.

\section{Analysis}\label{sec:analysis-single-content}
An operator, in order to optimize the offloading performance and cost, has to weigh its options and take decisions about: \textit{how to deliver} each content (directly or through offloading), \textit{how many copies} of a content should be placed to different edge nodes, \textit{which contents to store} in the SC and/or MN caches when their capacity is limited, etc. To this end, in this section, we provide the analytical results that are needed when trying to answer these questions. Specifically, we provide simple, closed form expressions for the performance of the "offloading on the edge" mechanism (Section~\ref{sec:single-generic}), and the costs it incurs (Section~\ref{sec:analysis-cost}).

\subsection{Content Dissemination}\label{sec:single-generic}
The performance  of the ``offloading on the edge'' mechanism depends on how much traffic it can offload and/or how fast are contents delivered. To answer these questions, we calculate the two main (and most common) performance metrics, namely the \textit{content delivery probability}, and \textit{content delivery delay}.

First, we state the following Lemma, in which we use a mean field approximation and a resulting system of ODEs to approximate the number of holders and requesters over time.

\begin{lemma}~\label{thm:ODEs}
The fluid-limit deterministic approximation for the expected number of holders ($H(t)$) and requesters ($R(t)$) at time $t$, is
\begin{align*}
H(t) &{=}  H_{0}\cdot \frac{(p_{c}\cdot R_{0}+H_{0})\cdot e^{\mu_{\lambda}\cdot(p_{c}\cdot R_{0}+H_{0})\cdot t }}{p_{c}\cdot R_{0}+H_{0}\cdot e^{\mu_{\lambda}\cdot(p_{c}\cdot R_{0}+H_{0})\cdot t }}\\
R(t) &{=} R_{0}\cdot \frac{p_{c}\cdot R_{0}+H_{0}}{p_{c}\cdot R_{0}+H_{0}\cdot e^{\mu_{\lambda}\cdot(p_{c}\cdot R_{0}+H_{0})\cdot t }} 
\end{align*}
where $H_{0} = H(0^{+})$ and $R_{0} = R(0^{+})$.
\end{lemma}
\begin{proof} Having assumed Poisson meeting processes, we can model the dissemination of a content with a continuous Markov Chain, whose states correspond to the different sets of holders and requesters $\{\mathcal{H},\mathcal{R}\}$. Fig.~\ref{fig:markov-chain} shows a segment of this Markov Chain; we present the different states with equal number of holders ($|\mathcal{H}|$) and requesters ($|\mathcal{R}|$) under the same group, which can be described by the tuples $\{|\mathcal{H}|,|\mathcal{R}|\}$. To transition between states a \textit{content delivery}, which takes place when a holder $i\in\mathcal{H}$ and a requester $j\in\mathcal{R}$ meet, is needed:
(i) \textit{Content delivery to cooperative node.} The next state is $\{|\mathcal{H}|=m+1,|\mathcal{R}|=n-1\}$ and the transition rate
\begin{equation}\label{eq:transition-rate-pc-generic}
 \lambda_{(m,n)\rightarrow(m+1,n-1)} = p_{c}\cdot \textstyle\sum_{i\in\mathcal{H}}\sum_{j\in\mathcal{R}}\lambda_{ij}
\end{equation}
(ii) \textit{Content delivery to non-cooperative node.} The next state is $\{|\mathcal{H}|=m,|\mathcal{R}|=n-1\}$ and the transition rate
\begin{equation}\label{eq:transition-rate-1-pc-generic}
 \lambda_{(m,n)\rightarrow(m,n-1)} = (1-p_{c})\cdot \textstyle\sum_{i\in\mathcal{H}}\sum_{j\in\mathcal{R}}\lambda_{ij}
\end{equation}

\begin{figure}
\includegraphics[width=1\linewidth, height=0.5\linewidth]{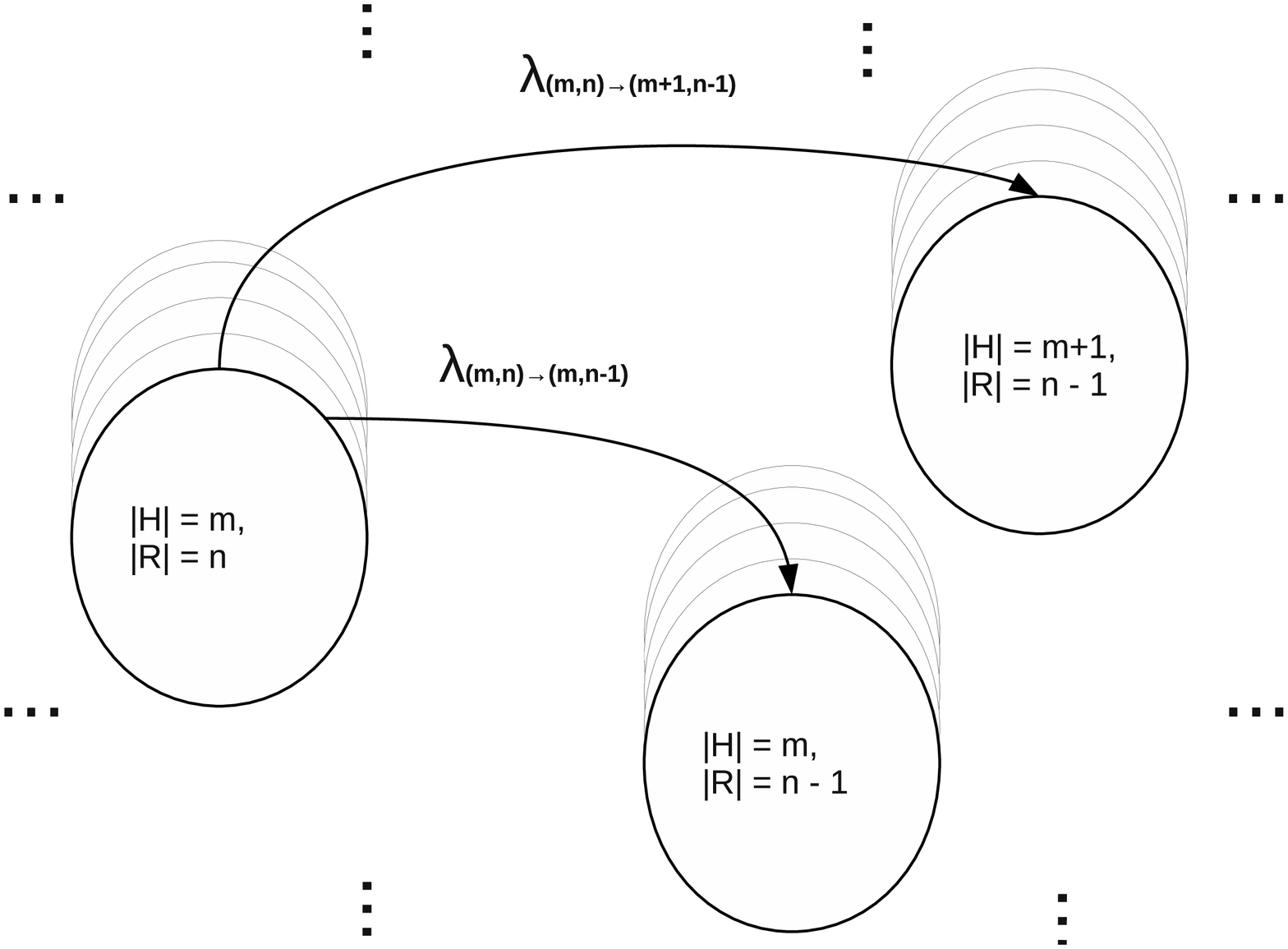}
\caption{Content dissemination modeled by a Markov Chain.}
\label{fig:markov-chain}
\end{figure}


Statistics for the content dissemination process over time (e.g. distribution of $|\mathcal{H}(t)|$ or $|\mathcal{R}(t)|$), can be computed using the transition matrix of the Markov Chain of Fig.~\ref{fig:markov-chain}. However, this would render the problem analytically (and numerically, for large networks) intractable. To this end, we approach the problem with a mean field approximation of stochastic reaction models~\cite{mean-field-toolbox}.

We first form the \textit{drift equation}~\cite[Theorem~1.4.1]{mean-field-toolbox} for the expected number of holders, $E\left[|\mathcal{H}(t)|\right]\equiv E\left[H(t)\right]$, as:
 \begin{align*}
 \frac{dE\left[H(t)\right]}{dt} &= E\left[\lambda_{(m,n)\rightarrow(m+1,n-1)}\right] = p_{c}\cdot E\left[\sum_{i\in\mathcal{H}}\sum_{j\in\mathcal{R}}\lambda_{ij}\right]
\end{align*}
The expectation in the right side of the drift equation is difficult to compute, as it requires the computation of the probabilities over the whole state space $\{\mathcal{H},\mathcal{R}\}$. To this end, one can approximate $E[H(t)]$ with its deterministic equivalent $h(t)$. This approximation comes after neglecting the variability of $H(t)$ around its mean value and becomes more accurate for larger systems~\cite[Section~1.5]{mean-field-toolbox}.

Based on the deterministic approximation and since (a) the rates $\lambda_{ij}$ are drawn independently from a distribution $f_{\lambda}(\lambda)$ with mean value $\mu_{\lambda}$ ($E[\lambda_{ij}]=\mu_{\lambda}$), and (b) the sum $\sum_{i\in\mathcal{H}}\sum_{j\in\mathcal{R}}\lambda_{ij}$ consists of $|\mathcal{H}|\cdot |\mathcal{R}|$ terms, we can write
\begin{equation}\label{eq:mean-value-approximation}
 \textstyle E\left[\sum_{i\in\mathcal{H}}\sum_{j\in\mathcal{R}}\lambda_{ij}\right] \approx h(t)\cdot r(t)\cdot \mu_{\lambda}
\end{equation}
The higher the number of terms in the above sum, and the less the heterogeneity of the meeting rates (i.e. the variance of $f_{\lambda}(\lambda)$), the more accurate the approximation in \eq{eq:mean-value-approximation} is.

Substituting \eq{eq:mean-value-approximation} in the drift equation (where $H(t)\rightarrow h(t)$), gives the ordinary differential equation (ODE) for $h(t)$\footnote{Note the differences between $H(t)$ and $h(t)$: (a) $H(t)$ is integer, whereas $h(t)$ is a real number; (b) the drift equation for $H(t)$ contains expectations, while the respective ODE for $h(t)$ does not.}
\begin{equation}\label{eq:ODE-h(t)}
 \frac{dh(t)}{dt} = p_{c}\cdot h(t)\cdot r(t)\cdot \mu_{\lambda}
\end{equation}

Proceeding similarly, the ODE for the deterministic approximation of the number of requesters ($R(t)\rightarrow r(t)$), is 
\begin{equation}\label{eq:ODE-r(t)}
 \frac{dr(t)}{dt} = -h(t)\cdot r(t)\cdot \mu_{\lambda}
\end{equation}

Finally, solving the system of the ODEs of \eq{eq:ODE-h(t)} and \eq{eq:ODE-r(t)}, gives the expressions of Lemma~\ref{thm:ODEs}.
\end{proof}

Based on Lemma~\ref{thm:ODEs} we, now, proceed to the calculation of the performance metrics. Let us consider a requester $i\in\mathcal{R}(0^{+})$, and denote as $T_{i}$ the time it receives the content. The probability this (random) requester to receive the content by a time $t$, i.e. $P\{T_{i}\leq t\}$, is equal to the \textit{percentage of offloaded contents} by time $t$. Hence, we can write
\begin{equation}\label{eq:definition-delivery-prob}
 P\{T_{i}\leq t\} = \frac{R_{0}-R(t)}{R_{0}} = 1-\frac{R(t)}{R_{0}}
\end{equation}
Substituting the expression of Lemma~\ref{thm:ODEs} in \eq{eq:definition-delivery-prob}, gives the following Result for the content delivery probability
%
%
%

\begin{result}[\textbf{Delivery Probability}]\label{result:delivery-probability-single}
The probability a content to be delivered to a requester by time $t$ is given by
\begin{equation*}
 P\{T_{d}\leq t\} =  1-\frac{p_{c}\cdot R_{0}+H_{0}}{p_{c}\cdot R_{0}+H_{0}\cdot e^{\mu_{\lambda}\cdot (p_{c}\cdot R_{0}+H_{0}) \cdot t}}
\end{equation*}
where $H_{0} = H(0^{+})$ and $R_{0} = R(0^{+})$.
\end{result}

With respect to the average delay a requester experiences till it receives the content, we state the following Result (the proof is technical and can be found in Appendix~\ref{appendix:proof-delay}). We derive expressions for two cases: (a) the content does not expire (i.e. $TTL\rightarrow\infty$), and (b) a macro-cell BS serves undelivered contents at time $t=TTL$. .
\begin{result}[\textbf{Delivery Delay}]\label{result:expected-delay-single}
The expected content delivery delay, under an expiry time $TTL\in[0,\infty)$, is given by\\
$-$ for $p_{c}>0$:
\begin{flalign*}
E[T_{d}|TTL] &=  \frac{\ln\left(1+\displaystyle\frac{p_{c}\cdot R_{0}-e^{-\mu_{\lambda}\cdot (p_{c}\cdot R_{0}+H_{0})\cdot TTL}}{H_{0}+p_{c}\cdot R_{0}\cdot e^{-\mu_{\lambda}\cdot (p_{c}\cdot R_{0}+H_{0})\cdot TTL}}\right)}{\mu_{\lambda}\cdot p_{c}\cdot R_{0}}&
\end{flalign*}
$-$ for $p_{c}=0$:
\begin{flalign*}
E[T_{d}|TTL] &=\frac{1-e^{-\mu_{\lambda}\cdot H_{0}\cdot TTL}}{\mu_{\lambda}\cdot H_{0}}&
\end{flalign*}
where $H_{0} = H(0^{+})$ and $R_{0} = R(0^{+})$.

%
\end{result}


\subsection{Content Delivery Cost}\label{sec:analysis-cost}
Incorporating the offloading costs (Section~\ref{sec:cost-model}) in our content dissemination model, and using the analytical results of Section~\ref{sec:single-generic}, we calculate the cost of a single content delivery in Result~\ref{thm:lemma-single-cost}. The expression we derive, gives the cost as a (simple) function of the system parameters (e.g. $R_{0}$, $\mu_{\lambda}$) and the operator selected parameters (e.g. $H_{SC}(0)$, $H_{MN}(0)$), providing, thus, the necessary information for the evaluation and tuning of the ``offloading on the edge'' mechanism.

\begin{result}\label{thm:lemma-single-cost}
The cost of ``offloading on the edge'' a content is given by
\begin{align*}
C =& C_{BH}\cdot H_{SC}(0)+ C_{BS}\cdot H_{MN}(0)\\
	&+ \left(C_{SC}\cdot q+ C_{D2D}\cdot(1-q) \right)\cdot  R_{0}	\cdot P\{T_{d}\leq TTL\}\\
	&+ C_{BS}^{(TTL)}\cdot R_{0} \cdot \left(1-P\{T_{d}\leq TTL\}\right)
\end{align*}
where $q = \textstyle \frac{H_{SC}(0)\cdot \ln\left(\frac{H(TTL)}{H_{0}}\right)}{p_{c}\cdot\left(R_{0}-R(TTL)\right)}$, and $P\{T_{d}\leq~TTL\}$, $H(TTL)$ and $R(TTL)$ are given in Lemma~\ref{thm:ODEs} and Result~\ref{result:delivery-probability-single}.
\end{result}
\begin{proof}~\\
$-$ {\textit{Initial Placement.}} The first two terms correspond to the initial placement phase: The cellular network operator, at time $t=0$, places the content to $H_{SC}(0)$ SCs and $H_{MN}(0)$ MNs; in total ($H_{0} = H_{SC}(0)+H_{MN}(0)$) holders. The costs per content placement are $C_{BH}$ and $C_{BS}$, respectively.\\

\noindent$-$ {\textit{Opportunistic Offloading.}} During the opportunistic offloading phase, i.e. $t\in(0,TTL)$, the average number of requesters that receive the content by an edge node is $R_{0}\cdot P\{T_{d}\leq TTL\}$. If we denote with $q$ the percentage of requesters that receive the content by a SC, it is easy to see that the costs due to SC-MN and MN-MN content deliveries are
\begin{align} 
C_{SC}\cdot q\cdot  R_{0}\cdot P\{T_{d}\leq TTL\}\label{eq:cost-SC}\\
C_{D2D}\cdot(1-q)\cdot  R_{0}\cdot P\{T_{d}\leq TTL\}\label{eq:cost-D2D}
\end{align}
respectively. 

To calculate the percentage $q$ we proceed as following:

At first, the total number of requesters that receive the content by time $TTL$ is
\begin{equation}\label{eq:tot-requesters-ttl}
\#R_{tot} = R_{0}-R(t)
\end{equation}

Second, the total number of requesters that receive the content in the interval $(t,t+dt]$, $t\in(0,TTL)$ is
\begin{equation}
 R(t)-R(t,t+dt) = -dR(t)
\end{equation}
The probability that a content delivery that takes place in the interval in the interval $(t,t+dt]$ is due to a SC is equal to
\begin{equation}
 \frac{H_{SC}(0)}{H(t)}\in[0,1]
\end{equation}
where $H_{SC}(0)$ is the number of SC holders (which does not change over time), and $H(t)$ the total number of holders at time $t$.

Therefore, the number of requesters that receive the content by an SC holder in the interval $(t,t+dt]$ is given by $-dR(t)\cdot \frac{H_{SC}(0)}{H(t)}$, and the total number of requesters that receive the content by an SC holder by time $TTL$ is
\begin{align}
\#R_{SC}
	&=\textstyle\int_{0}^{TTL} -dR(t)\cdot \frac{H_{SC}(0)}{H(t)} =  \int_{0}^{TTL} -\frac{dR(t)}{dt}\cdot \frac{H_{SC}(0)}{H(t)} \cdot dt\nonumber\\
 	&\stackrel{\text{\eq{eq:ODE-r(t)}}}{=} \textstyle\int_{0}^{TTL} H(t)\cdot R(t)\cdot \mu_{\lambda}\cdot \frac{H_{SC}(0)}{H(t)} \cdot dt\nonumber\\
 	&= \mu_{\lambda}\cdot H_{SC}(0)\textstyle\int_{0}^{TTL} R(t)\cdot dt
\end{align}
Using the expression of Lemma~\ref{thm:ODEs} for $R(t)$ to calculate the above integral, we get
\begin{align}\label{eq:sc-requesters-ttl}
\#R_{SC}
 	&= \frac{H_{SC}(0)}{p_{c}} \cdot \ln\left(\frac{(p_{c}\cdot R_{0}+H_{0})\cdot e^{\mu_{\lambda}\cdot(p_{c}\cdot R_{0}+H_{0})\cdot TTL }}{p_{c}\cdot R_{0}+H_{0}\cdot e^{\mu_{\lambda}\cdot(p_{c}\cdot R_{0}+H_{0})\cdot TTL }}\right)\nonumber\\ 	
 	&= \frac{H_{SC}(0)}{p_{c}} \cdot \ln\left(\frac{H(TTL)}{H_{0}}\right)
\end{align}
where the last equality follows from the expression for $H(t)$ given in Lemma~\ref{thm:ODEs}.

Now, $q$ easily follows from \eq{eq:tot-requesters-ttl} and \eq{eq:sc-requesters-ttl}
\begin{equation}
q = \frac{\#R_{SC}}{\#R_{tot}} = \frac{H_{SC}(0)}{p_{c}} \cdot \frac{\ln\left(\frac{H(TTL)}{H_{0}}\right)}{R_{0}-R(TTL)}
\end{equation}

\noindent$-$ {\textit{Delayed Delivery.}} Finally, the average number of requesters that do not receive the content before its expiry time, is given by $R_{0}\cdot \left(1-P\{T_{d}\leq TTL\}\right)$. Since, the cost of each content transmission at time $t=TTL$ is $C_{BS}^{(TTL)}$, the total cost of delayed delivery phase (last line of the expression in Lemma~\ref{thm:lemma-single-cost}) follows easily.
\end{proof}

\section{Applications: Cost Optimization}\label{sec:single-cost-optimization}
In a real scenario, the network operator would have to offload simultaneusly many different contents. Using the results of the previous section, the average performance or the total cost over all the contents can be calculated easily, by evaluating them for each content separately and then averaging or summing them, respectively. However, since some of the system parameters are controlled by the operator (e.g. $H_{0}$), they can be selected such that they lead to optimal performance. To this end, in this section, as an application of our analytical results, we study how offloading and caching can be designed in order to minimize the total cost.

\subsection{Optimizing the Total Offloading Cost}\label{sec:optimizing-total-cost}
Let us assume that the content provider has to deliver $M\geq1$ contents to their requesters. We denote the set of the contents as $\mathcal{M}$ ($M=|\mathcal{M}|$).  Since in a real scenario, not all contents are expected to be equally popular~\cite{youtube-traffic-from-edge,top-video-cellular,pptv-mobile-vod}, nor tolerate equal delays, we denote the popularity (i.e. the number of initial requesters) and the expiry time of each content $\theta\in\mathcal{M}$ as $R_{0}^{\theta}$ and $TTL^{\theta}$, respectively.

Under a given setting (i.e. with certain mobility, cooperation, traffic, etc., characteristics), what the cellular network can select, is the initial placement (\textit{caching}) for each content $\theta\in\mathcal{M}$; namely, the number of SC and MN initial holders, $H_{SC}^{\theta}(0)$ and $H_{MN}^{\theta}(0)$, respectively (note that $H_{0}^{\theta}(0)\equiv H_{SC}^{\theta}(0)+H_{MN}^{\theta}(0)$).

Therefore, if we denote as $C^{\theta}$ is the delivery cost of a content $\theta\in\mathcal{M}$ (which is given by Result~\ref{thm:lemma-single-cost}), we can express the \textit{total} cost optimization problem as
\begin{problem}\label{eq:optimization-problem-multi}~\\
$~\hspace{1.6cm}\min_{\overline{H}_{SC},~\overline{H}_{MN}}~\left\{\sum_{\theta\in\mathcal{M}} C^{\theta}\right\}$
\begin{align*}
~s.t.~~&\forall\theta\in\mathcal{M}:~  0\leq H_{SC}^{\theta}(0)\leq N_{SC}\\
&\hspace{1.65cm} 0\leq H_{MN}^{\theta}(0)\leq R^{\theta}(0)\\
&\text{and}\hspace{1cm }\sum_{\theta\in\mathcal{M}}H_{SC}^{\theta}(0)\leq \sum_{i\in\mathcal{SC}} Q(i)
\end{align*}
where $\overline{H_{SC}}$ and $\overline{H_{MN}}$ denote the vectors with components $H_{SC}^{\theta}(0)$ and $H_{MN}^{\theta}(0)$ ($\theta\in\mathcal{M}$), respectively, and $Q(i)$ is the caching capacity (in number of contents) of a SC node\footnote{Since MNs cache only contents in which they are interested in, we assume that their storage capacity is enough for all the contents of interest. Hence, storage capacity constraints for MNs are not considered in Problem~\ref{eq:optimization-problem-multi}.}~$i$.
\end{problem}

The costs $C^{\theta}$ in the objective function of Problem~\ref{eq:optimization-problem-multi} can be expressed as a function of the optimization variables (Result~\ref{thm:lemma-single-cost}). As a result, well known numerical methods can be employed to solve Problem~\ref{eq:optimization-problem-multi}. In the following numerical example, we calculate the optimal content placement by solving Problem~\ref{eq:optimization-problem-multi}, apply the ``offloading on the edge'' mechanism, and demonstrate how it can reduce the total content delivery cost.

\textbf{\textit{Numerical Example:}} We assume a network with $N_{SC}=4$ SCs, each of them having a storage capacity of $Q=100$ contents. Edge nodes meet each other with an average meeting rate $\mu_{\lambda} = 3.3\cdot 10^{-5}$ meetings/sec (equal to this of the real mobility trace~\cite{infocom-trace}). The cellular network has to deliver $M=100$ contents of equal size and expiry time $TTL$. Content popularity is given by a bounded Pareto distribution in the interval $R_{0}\in[10, 1000]$ with shape parameter $\alpha = 0.5$~\cite{youtube-traffic-from-edge} or $\alpha=1$. The (normalized) costs are\footnote{In general, the offloading costs incurred in each phase, might differ between areas, time periods and operators. Their absolute values are not available and/or are difficult to estimate. To this end, in this example, as well as in other numerical results, we use relative values inferred by some average values proposed in~\cite{johansson2007cost}.} 
\begin{align*}
 C_{BH} &= 0.8~~~~~&C_{SC}&=0.2~~~~~&C_{BS}^{(TTL)}=2\\
 C_{BS}&=1 ~~~~~&C_{D2D}&=0.1~~~~~&
\end{align*}

Fig.~\ref{fig:numerical-example-cost-generic} shows how the total content delivery cost decreases when the cellular network offloads contents on the edge nodes. Specifically, we present the \textit{Relative Cost Decrease}, \textit{RCD}, $\frac{C-C_{off}}{C}$, where $C$ and $C_{off}$ denote the total cost of delivering the contents without and with offloading, respectively. High values of the \textit{RCD} correspond to cases where the cost of ``offloading on the edge'' $C_{off}$ is small (compared to the cost without offloading $C$), whereas low values of the \textit{RCD} indicate that the gain due to offloading is not significant.

In Fig.~\ref{fig:numerical-example-cost-generic-vs-TTL} we present how \textit{RCD} changes for different values of delay tolerance, $TTL$, when the percentage of MNs contributing to offloading is $10\%$ ($p_{c}=0.1$). In Fig.~\ref{fig:numerical-example-cost-generic-vs-pc} we present the change of \textit{RCD} for different MN cooperation levels ($p_{c}$) and for a delay tolerance equal to $TTL=5min$. The two curves, in both figures, correspond to scenarios where content popularity is Pareto distributed with exponent $\alpha=0.5$ and $\alpha = 1$. Some important observations are the following: (i) Higher values $p_{c}$ or $TTL$ lead to lower total offloading costs, i.e. the more willing an MN is to offload contents or to tolerate delays, the more effective the ``offloading on the edge'' becomes. In particular, for the presented scenarios, the content delivery cost can decrease up to $80\%$ for large $TTL$ values (Fig.~\ref{fig:numerical-example-cost-generic-vs-TTL}), while for large $p_{c}$ values the offloading cost can be $40\%$ less than the cost without offloading (Fig.~\ref{fig:numerical-example-cost-generic-vs-pc}). (ii) The cost gains are larger for smaller Pareto exponents $\alpha$. This is due to the fact that (a) for the more skewed distributions, i.e. $\alpha=1$, there exist more unpopular contents than in the case of $\alpha=0.5$, and (b) ``offloading on the edge'' (and, in general, offloading mechanisms) is more efficient when distributing popular contents.

\begin{figure}
\subfigure[$p_{c}=0.1$]{\includegraphics[width=0.49\linewidth]{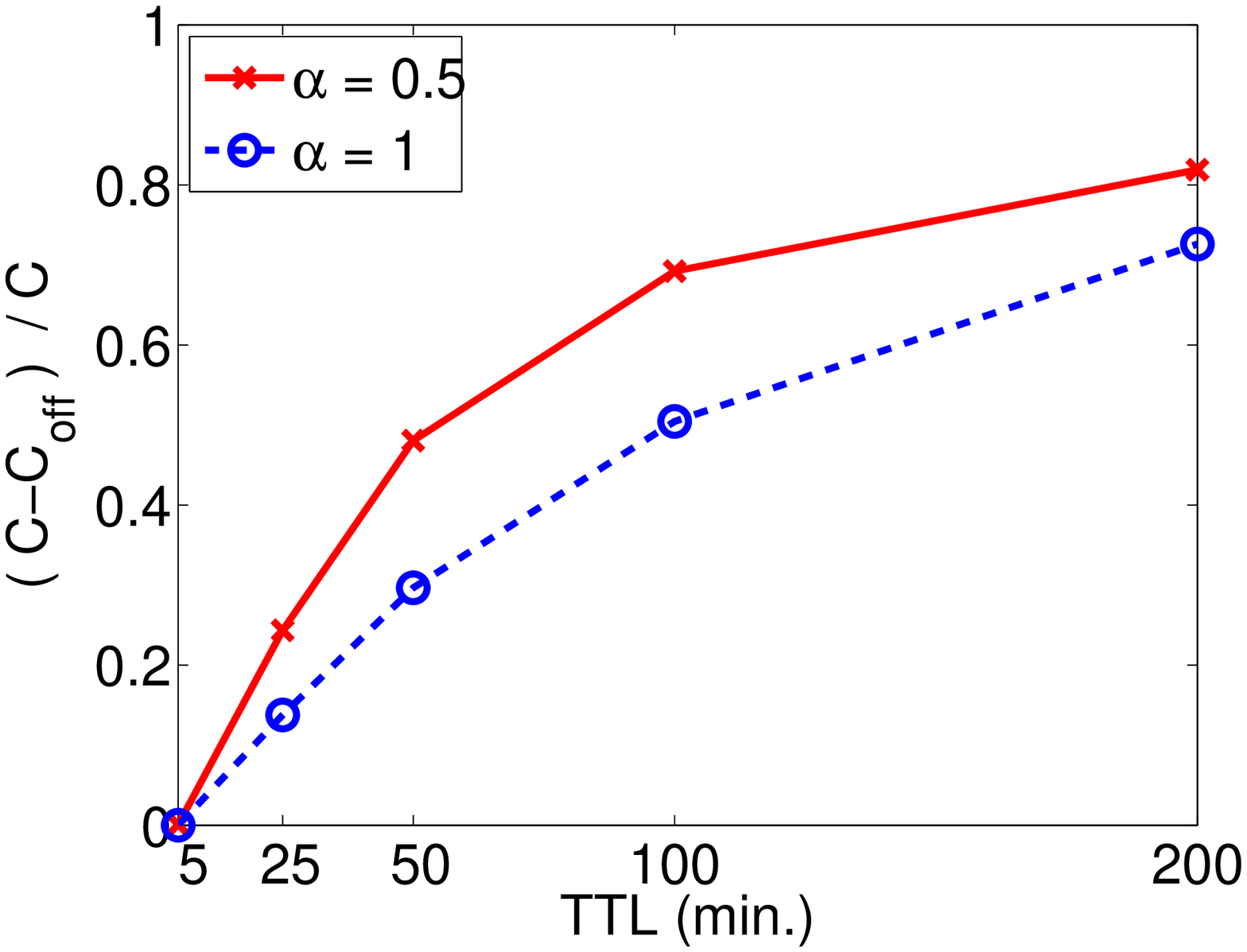}\label{fig:numerical-example-cost-generic-vs-TTL}}
\subfigure[$TTL = 5min.$]{\includegraphics[width=0.49\linewidth]{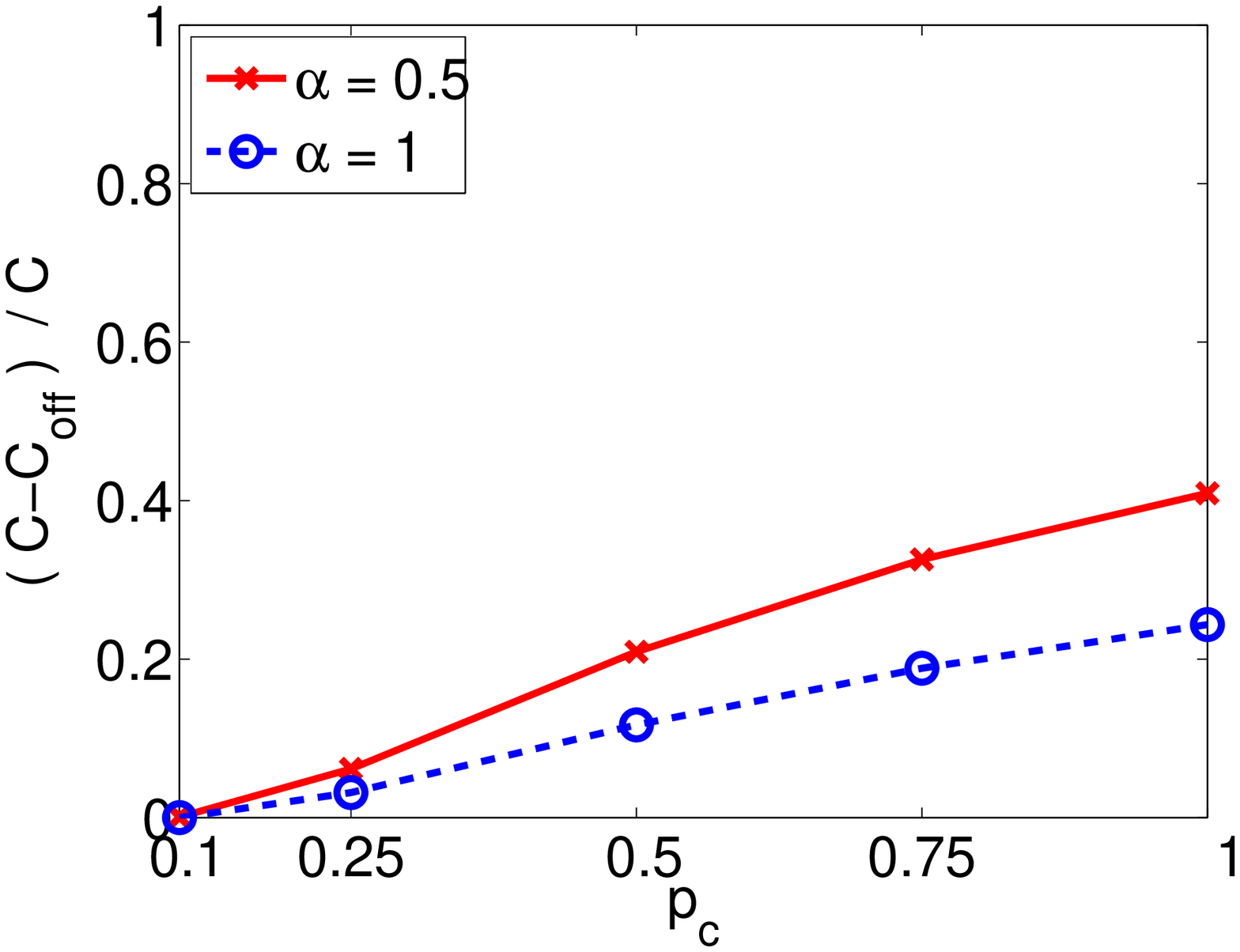}\label{fig:numerical-example-cost-generic-vs-pc}}
\caption{\textit{Relative Cost Decrease} $\frac{C-C_{off}}{C}$, where $C$ and $C_{off}$ denote the total cost of delivering the contents without and with offloading, respectively.}
\label{fig:numerical-example-cost-generic}
\end{figure}

\subsection{Case Studies: Optimal Content Placement}

Under certain scenarios, analytical solutions for Problem~\ref{eq:optimization-problem-multi} can be found as well. In the remainder, we focus on two characteristic cases, which are analytically solvable, and provide useful insights for the system.

\subsection*{\underline{Offloading through SCs}}
We first consider the case where contents are offloaded only through SCs (i.e. when $p_{c}=0$ and $H_{MN}^{\theta}(0)=0$, or equivalently, $H_{0}^{\theta} = H_{SC}^{\theta}(0)$). This is the most common and feasible scenario considered in previous literature, since MNs are not required to share their contents, and thus incentive mechanisms are easier to implement. In this case and for\footnote{The ``offloading on the edge'' mechanism is meaningful if $C_{SC}<C_{BS}^{(TTL)}$, as in the opposite case, offloading would cost more than directly delivering from the macro-cell BSs.} $C_{SC}<C_{BS}^{(TTL)}$ it can be proved that Problem~\ref{eq:optimization-problem-multi} is convex and we compute the analytical solution in Result~\ref{THM:OPTIMAL-H0-NO-COOP-MULTIPLE}. For notation simplicity, we consider equal expiry times $TTL^{\theta}=TTL, ~\forall\theta\in\mathcal{M}$, and cache sizes $Q(i)=Q,~\forall i\in\mathcal{SC}$. However, Result~\ref{THM:OPTIMAL-H0-NO-COOP-MULTIPLE} can be easily modified for different\footnote{In particular, one has to substitute $\gamma$ with $\gamma^{\theta}=\mu_{\lambda}\cdot TTL^{\theta}$ for each content. The expressions for $H_{SC}^{\theta}(0)$ remain the same, and only the expressions of $L$ and $U$ need to be modified.} $TTL^{\theta}$ and $Q(i)$ values.

\begin{result}\label{THM:OPTIMAL-H0-NO-COOP-MULTIPLE}
Under a base scenario ($p_{c}=0$, $H_{MN}(0)=0$), the initial allocation $\overline{H_{SC}}$ that minimizes the total cost, is given by
\begin{equation*}
 H_{SC}^{\theta}(0) = \left\{
 \begin{array}{lc}
  N_{SC}& ,  R^{\theta}(0)>  U\\
 \frac{1}{\gamma}\cdot\ln\left(\frac{1}{L}\cdot R^{\theta}(0)\right)& , L\leq R^{\theta}(0)\leq U\\
  0		& , R^{\theta}(0)< L\\
 \end{array}
 \right.
 \end{equation*}
with  $\gamma = \mu_{\lambda}\cdot TTL$, $L = \frac{1}{\gamma\cdot\Phi}\cdot \left(1+\frac{\lambda_{0}}{C_{BH}}\right)$, $U =L\cdot e^{\gamma\cdot N_{SC}}$, $\Phi = \frac{C_{BS}^{(TTL)}-C_{SC}}{C_{BH}}$, and 
\begin{equation*}
\lambda_{0} = \inf \left\{\lambda_{0}\geq 0: \sum_{\theta\in\mathcal{M}} H_{SC}^{\theta}(0)\leq\sum_{i\in\mathcal{SC}} Q(i)\right\}
\end{equation*}
\end{result}
\begin{proof}
Applying the method of Lagrange multipliers~\cite{practical-optimization-book} to Problem~\ref{eq:optimization-problem-multi}, gives (for brevity we use the notation $H_{0}^{\theta}\equiv~H_{SC}^{\theta}(0^{+})=H_{SC}^{\theta}(0)$ and $R_{0}^{\theta}\equiv R^{\theta}(0^{+})=R^{\theta}(0)$):
\begin{multline}\label{eq:lagrangian-equation}
 \nabla \left(\sum_{\theta\in\mathcal{M}} C^{\theta}\right) = 
	\nabla \lambda_{0} \left(\sum_{i\in\mathcal{SC}} Q(i)-\sum_{\theta\in\mathcal{M}} H_{0}^{\theta}\right)\\
	+ \nabla \sum_{\theta\in\mathcal{M}}\lambda_{\theta} \cdot H_{0}^{\theta} +\nabla \sum_{\theta\in\mathcal{M}}\mu_{\theta}\cdot(N_{SC}-H_{0}^{\theta})
\end{multline}
where $\lambda_{0}\geq0$ and $\lambda_{\theta},\mu_{\theta}\geq0,\forall\theta\in\mathcal{M}$ are the langrangian multipliers.

Using the expression of Result~\ref{result:delivery-probability-single} for the delivery probability, the offloading cost (Result~\ref{thm:lemma-single-cost}) of a content $\theta$, in a base scenario, can be written as
\begin{equation}\label{eq:cost-single-for-optimization}
C^{\theta} = C_{BH}\cdot H_{0}^{\theta} + C_{SC}\cdot R_{0}^{\theta}+(C_{BS}^{(TTL)}-C_{SC})\cdot R_{0}^{\theta}\cdot e^{-\mu_{\lambda}\cdot H_{0}^{\theta} \cdot TTL}
\end{equation}

Substituting $C^{\theta}$ from \eq{eq:cost-single-for-optimization} to \eq{eq:lagrangian-equation}, the differentiation over $H_{0}^{\theta}$ gives
\begin{equation}\label{eq:optimal-Ho-langragian}
 H_{0}^{\theta} = \frac{1}{\gamma}\cdot \left[\ln\left(\Phi\cdot \gamma\cdot R_{0}^{\theta}\right)-\ln\left(1+\frac{\lambda_{0}-\lambda_{\theta}+\mu_{\theta}}{C_{BH}}\right)\right]
\end{equation}

\noindent The conditions for the lagrangian multipliers, i.e.
\begin{align*}
\lambda_{\theta} \cdot H_{0}^{\theta}&=0,~~~\text{and}~~~~\mu_{\theta}\cdot(N_{SC}-H_{0}^{\theta})=0~~~~,\forall\theta\in\mathcal{M}
\end{align*}
imply that $H_{0}^{\theta}$ either 
\begin{itemize}
\item[(a)] is given by \eq{eq:optimal-Ho-langragian} and $\lambda_{\theta}=\mu_{\theta}=0$, or
\item[(b)] is equal to $N_{SC}$ and $\lambda_{\theta}=0, \mu_{\theta}>0$, or 
\item[(c)] is equal to $0$ and $\lambda_{\theta}>0, \mu_{\theta}=0$ 
\end{itemize}

From condition (a), we calculate the limits of the interval within which the optimal $H_{0}^{\theta}$ is given by \eq{eq:optimal-Ho-langragian}. To find the lower limit, $L$, we set $H_{0}^{\theta}$ (\eq{eq:optimal-Ho-langragian} with $\lambda_{\theta}=\mu_{\theta}=0$) %
 equal to $0$ and for the upper limit, $U$, equal to $N_{SC}$, which give
\begin{subequations}\label{eqs:L-U}
 \begin{align}
 L &= \frac{1}{\gamma\cdot\Phi}\cdot\left(1+\frac{\lambda_{0}}{C_{BH}}\right)		\label{eq:L-low-limit}\\
U &= \frac{1}{\gamma\cdot\Phi}\cdot e^{\gamma\cdot N_{SC}}\cdot\left(1+\frac{\lambda_{0}}{C_{BH}}\right)		\label{eq:U-upper-limit}= L\cdot e^{\gamma\cdot N_{SC}}
\end{align}
\end{subequations}

Combining \eq{eq:optimal-Ho-langragian} and \eqs{eqs:L-U}, we can express the optimal placement as
\begin{equation}\label{eq:optimal-Ho}
 H_{0}^{\theta~*} = \left\{
 \begin{array}{lc}
  N_{SC}& ,  R_{0}^{\theta}>  U\\
 \frac{\ln\left(\gamma\cdot\Phi\cdot R_{0}^{\theta}\right)-\ln\left(1+\frac{\lambda_{0}}{C_{BH}}\right)}{\gamma}& , L\leq R_{0}^{\theta}\leq U\\
  0		& , R_{0}^{\theta}< L\\
 \end{array}
 \right.
 \end{equation}

The only unknown parameter in \eq{eq:optimal-Ho} is $\lambda_{0}$ (since we expressed $L$ and $U$ as functions of $\lambda_{0}$). Lemma~\ref{thm:monotonicity}, which we state and prove in Appendix~\ref{appendix:lemma-monotonicity}, suggests that the total cost, $\sum_{\theta\in\mathcal{M}}C^{\theta}$, is monotonically increasing with $\lambda_{0}$. Therefore, the optimal placement policy corresponds to the smaller \textit{non-negative} value of $\lambda_{0}$ that satisfies the storage constraint, $\sum_{\theta\in\mathcal{M}} H_{0}^{\theta}\leq \sum_{i\in\mathcal{SC}} Q(i)$, and this proves the Result
.
\end{proof}
In general, the value of the parameter $\lambda_{0}$ can be found (within some precision) with e.g. a binary search. Nevertheless, for a large number of contents, and given their popularity distribution, its value can be directly calculated using the Corollary~\ref{thm:corollary}, which follows after substituting the expression of Result~\ref{THM:OPTIMAL-H0-NO-COOP-MULTIPLE} and the popularity density function in the storage constraint $\sum_{\theta\in\mathcal{M}} H_{SC}^{\theta}(0) = \sum_{i\in\mathcal{SC}} Q(i)$.
\begin{corollary}\label{thm:corollary}
Under a content popularity distribution $\rho(x)$, 
the parameter $\lambda_{0}$ in Result~\ref{THM:OPTIMAL-H0-NO-COOP-MULTIPLE} is given by $\lambda_{0} = \max\left\{0,\hat{\lambda_{0}}\right\}$, where $\hat{\lambda_{0}}$ is the (minimum) solution of
\begin{multline*}
\int_{L}^{U}\ln\left(\gamma\cdot\Phi\cdot x\right)\cdot \rho(x)dx - \ln\left(1+\frac{\lambda_{0}}{C_{BH}}\right)\cdot \int_{L}^{U}\rho(x)dx \\
	+\gamma\cdot N_{SC}\cdot \int_{U}^{\infty} \rho(x)dx = \frac{\gamma\cdot N_{SC}\cdot Q}{M}
\end{multline*}
\end{corollary}

Result~\ref{THM:OPTIMAL-H0-NO-COOP-MULTIPLE} reveals how resources should be allocated:
(i) The optimal allocation is logarithmic in popularity, with either large or small caches. (ii) When capacity is limited, an extra factor ($\lambda_{0}$) is introduced, so that the \emph{relative} allocation remains logarithmic, but the absolute allocation is reduced (normalized) as the number of contents increase, or total capacity decreases. (iii) Some low popularity contents might get no allocation, either because it does not help the offloading cost, or because there is not enough capacity for them.


\textbf{\textit{Practical Example:}} Assume an urban area covered by $N_{BS}=4$ macro-cell BSs and $N_{SC}=100$ SCs. On average, in this area reside $N_{MN}=10000$ users\footnote{Vodafone Germany reported an average number of $1700$ users per cell (\url{http://mobilesociety.typepad.com/mobile_life/2009/06/base-station-numbers.html}). In an urban environment, users density is expected to be higher.} with an average meeting rate $\mu_{\lambda} = 3.3\cdot 10^{-5}$ meetings/sec~\cite{infocom-trace}. The cellular network has to deliver $M$ contents (e.g. YouTube video files of an average size $10MB$~\cite{youtube-traffic-from-edge}) with expiry time $TTL\approx 5min$ and popularity given by a bounded Pareto distribution in the interval $R_{0}\in[10, 1000]$ with shape parameter $\alpha = 0.5$~\cite{youtube-traffic-from-edge}. The costs are $C_{BS}^{(TTL)} = 10\cdot C_{BH}$ and $C_{SC}\ll C_{BH},C_{BS}^{(TTL)}$.

Substituting the given values, and taking the expectation over the popularity distribution, it follows that the necessary buffer size of a SC, $Q =\frac{E[H_{0}]}{N_{SC}}\cdot M\cdot L$, is approximately $1 MB$ per content. This means that, even under very high traffic demand, the caching capacity of the SCs would be adequate such that the last constraint of Problem~\ref{eq:optimization-problem-multi} is not violated; e.g. for $M=100000$ (i.e. each user requests $10$ videos per $5$ minutes!), the needed capacity is $Q=100GB$ (which is a feasible and relatively cheap investment). 

\subsection*{\underline{Offloading through MNs}} 
We now consider the case where offloading takes place only through MN-MN communication ($p_{c}>0$) and \textit{without} content storing on SCs (i.e. $H_{SC}(0)=0$). A content is initially sent by the BSs to $H_{MN}(0)$ (out of $R(0)$) of its requesters, which start disseminating it to the other requesters. However, not all nodes might be willing to participate by acting as holders, which in our framework means that each node (including the initial nodes in which the content is placed) cooperates with probability $p_{c}$. Therefore, we can write 
\begin{align*}
 H_{0}\equiv H_{MN}(0^{+})= p_{c}\cdot H_{MN}(0)
\end{align*}
Also, as defined in Lemma~\ref{thm:ODEs},
\begin{align*}
 R_{0}\equiv R(0^{+}) = R(0)-H_{MN}(0)
\end{align*}

As in the previous case, we assume equal expiry times $TTL^{\theta}=TTL, ~\forall\theta\in\mathcal{M}$.

\begin{result}\label{THM:OPTIMAL-H0-MN-MN}
Under an \textit{opportunistic MN-MN scenario} ($p_{c}>0$, $H_{MN}(0)=0$), the initial allocation $\overline{H_{MN}}$ that minimizes the total cost, is given by
\begin{equation*}
 H_{MN}^{\theta}(0) = \left\{
 \begin{array}{lc}
  R^{\theta}(0)& ,  R^{\theta}(0)\leq  OPT^{\theta}\\
 OPT^{\theta} & , 0\leq OPT^{\theta}< R^{\theta}(0)\\
  0		& , OPT^{\theta}< 0\\
 \end{array}
 \right.
 \end{equation*}
where $\displaystyle{OPT^{\theta} =\frac{R^{\theta}(0)\cdot \left(\sqrt{\Phi^{'}}\cdot e^{\frac{1}{2}\gamma\cdot p_{c}\cdot R^{\theta}(0)}-1\right)}{e^{\gamma\cdot p_{c}\cdot R^{\theta}(0)}-1}}$, and $\Phi^{'}~=~\frac{C_{BS}^{(TTL)}-C_{D2D}}{C_{BS}-C_{D2D}}$ and $\gamma = \mu_{\lambda}\cdot TTL$.
\end{result}
\begin{proof}
The cost for offloading a content $\theta$ under an opportunistic MN-MN scenario, where $H_{0}^{\theta} = p_{c}\cdot H_{MN}^{\theta}(0)$ and $R_{0}^{\theta} = R(0)^{\theta}-H_{MN}^{\theta}(0)$, is (see Result~\ref{thm:lemma-single-cost})
\begin{multline}\label{eq:single-cost-MN-MN}
 C^{\theta} = C_{BS}\cdot H_{MN}^{\theta}(0)\\+ \left(C_{D2D}-C_{BS}^{(TTL)} \right)\cdot(R^{\theta}(0)-H_{MN}^{\theta}(0))\cdot  P\{T_{d}\leq TTL\} \\+ C_{BS}^{(TTL)} \cdot(R^{\theta}(0)-H_{MN}^{\theta}(0))\hfill~
\end{multline}
Similarly, for $H_{0}^{\theta} = p_{c}\cdot H_{MN}^{\theta}(0)$ and $R_{0}^{\theta} = R^{\theta}(0)-H_{MN}^{\theta}(0)$, the delivery probability $P\{T_{d}\leq TTL\}$ can be written as
\begin{equation}\label{eq:P-TTL-MN-MN}
 P\{T_{d}\leq TTL\} = 1- \textstyle\frac{R^{\theta}(0)}{R^{\theta}(0)+H_{MN}^{\theta}(0)\cdot \left(e^{\gamma \cdot p_{c}\cdot R^{\theta}(0)}-1\right)}
\end{equation}
where $\gamma = \mu_{\lambda}\cdot TTL$.

Substituting \eq{eq:P-TTL-MN-MN} in \eq{eq:single-cost-MN-MN}, and taking the derivative over the initial number of transmissions $H_{MN}^{\theta}(0)$, gives
\begin{multline}\label{eq:derivative-cost-MN-MN}
 \frac{dC^{\theta}}{dH_{MN}^{\theta}(0)} = (C_{BS}^{(TTL)}-C_{D2D})\\+\frac{(C_{D2D}-C_{BS})\cdot(R^{\theta}(0))^{2}\cdot e^{\gamma\cdot p_{c}\cdot R^{\theta}(0)}}{\left(R^{\theta}(0)+H_{MN}^{\theta}(0)\cdot (e^{\gamma\cdot p_{c}\cdot R^{\theta}(0)}-1)\right)^{2}}
\end{multline}
From \eq{eq:derivative-cost-MN-MN} it follows that 
\begin{equation*}
 \frac{dC^{\theta}}{dH_{MN}^{\theta}(0)} = \left\{
 \begin{array}{lc}
  <0 & ,  H_{MN}^{\theta}(0)<  OPT^{\theta}\\
  >0& ,  H_{MN}^{\theta}(0)>  OPT^{\theta}
 \end{array}
 \right.
 \end{equation*}
where
\begin{equation}
 OPT^{\theta} =\frac{R^{\theta}(0)\cdot \left(\sqrt{\Phi^{'}}\cdot e^{\frac{1}{2}\gamma\cdot p_{c}\cdot R^{\theta}(0)}-1\right)}{e^{\gamma\cdot p_{c}\cdot R^{\theta}(0)}-1}
\end{equation}
Therefore, when $OPT^{\theta}\in[0,R^{\theta}(0)]$, the minimum cost is achieved for $H_{MN}^{\theta}(0) = OPT^{\theta}$. Otherwise, for $OPT^{\theta}\notin[0,R^{\theta}(0)]$, and since it must hold that $H_{MN}^{\theta}(0)\in[0,R^{\theta}(0)]$, the minimum cost is achieved for the largest or lowest possible values of $H_{MN}^{\theta}(0)$.
\end{proof}

Result~\ref{THM:OPTIMAL-H0-MN-MN} reveals how content storage should be delivered when offloading only through MNs is considered. As it can be seen, the initial allocation is much different that in the offloading through SCs case (see Result~\ref{THM:OPTIMAL-H0-NO-COOP-MULTIPLE}), and this is mainly due to the fact that some of the requesters get the content at the beginning.

\section{Simulation Results}\label{sec:single-validation}
To validate our analysis, we compare the theoretical predictions against Monte Carlo simulations (Section~\ref{sec:model-validation}). Then, we evaluate the cost efficiency of "offloading on the edge" in scenarios with realistic traffic demand patterns (Section~\ref{sec:simulation-cost-efficiency}).

\subsection{Model Validation}\label{sec:model-validation}

\subsubsection{Synthetic Scenarios}
We first compare the theoretical results against Monte Carlo simulations on various synthetic scenarios. Synthetic simulations allow us to create a number of different scenarios with varying parameters.

We generate synthetic networks, conforming to the model of Section~\ref{sec:analysis-preliminaries}, as following:\\
(i) We choose a probability distribution $f_{\lambda}(\lambda)$ and for each pair $\{i,j\}$ we draw randomly a meeting rate $\lambda_{ij}$.\\
(ii) We create a sequence of contact events for every pair in the network with rate (Poisson processes with rates $\lambda_{ij}$).\\
(iii) We select the content traffic parameters ($R_{0}$, $H_{0}$, $p_{c}$, $H_{SC}(0)$, $H_{MN}(0)$, $N_{SC}$), and we simulate a large number of content disseminations, choosing randomly each time the set of requesters and the set of holders (note, however, that the set of holders depends also on the parameters $H_{SC}(0)$, $H_{MN}(0)$ and $N_{SC}$).

We have created many scenarios with different combinations of mobility ($f_{\lambda}(\lambda)$) and traffic ($R_{0}$, $H_{0}$, $p_{c}$, $H_{SC}(0)$, $H_{MN}(0)$, $N_{SC}$) characteristics. We present here a representative subset of them, which allow us demonstrate the accuracy of our predictions and their sensitivity when varying certain parameters. In the presented scenarios we create nodes mobility according to a gamma distribution $f_{\lambda}(\lambda)$ with mean value $\mu_{\lambda}=1$ (i.e. normalized value) and variance $\sigma_{\lambda}^{2}$ (or, equivalently, coefficient of variation $CV_{\lambda} = \frac{\sigma_{\lambda}}{\mu_{\lambda}}$)~\cite{passarella-dunbar-journal}. Gamma distributions allow us to capture different levels of mobility heterogeneity by varying the value of $CV_{\lambda}$.

\begin{figure}
\subfigure[$H_{0}=1$, $p_{c}=0.5$]{\includegraphics[width=0.49\linewidth]{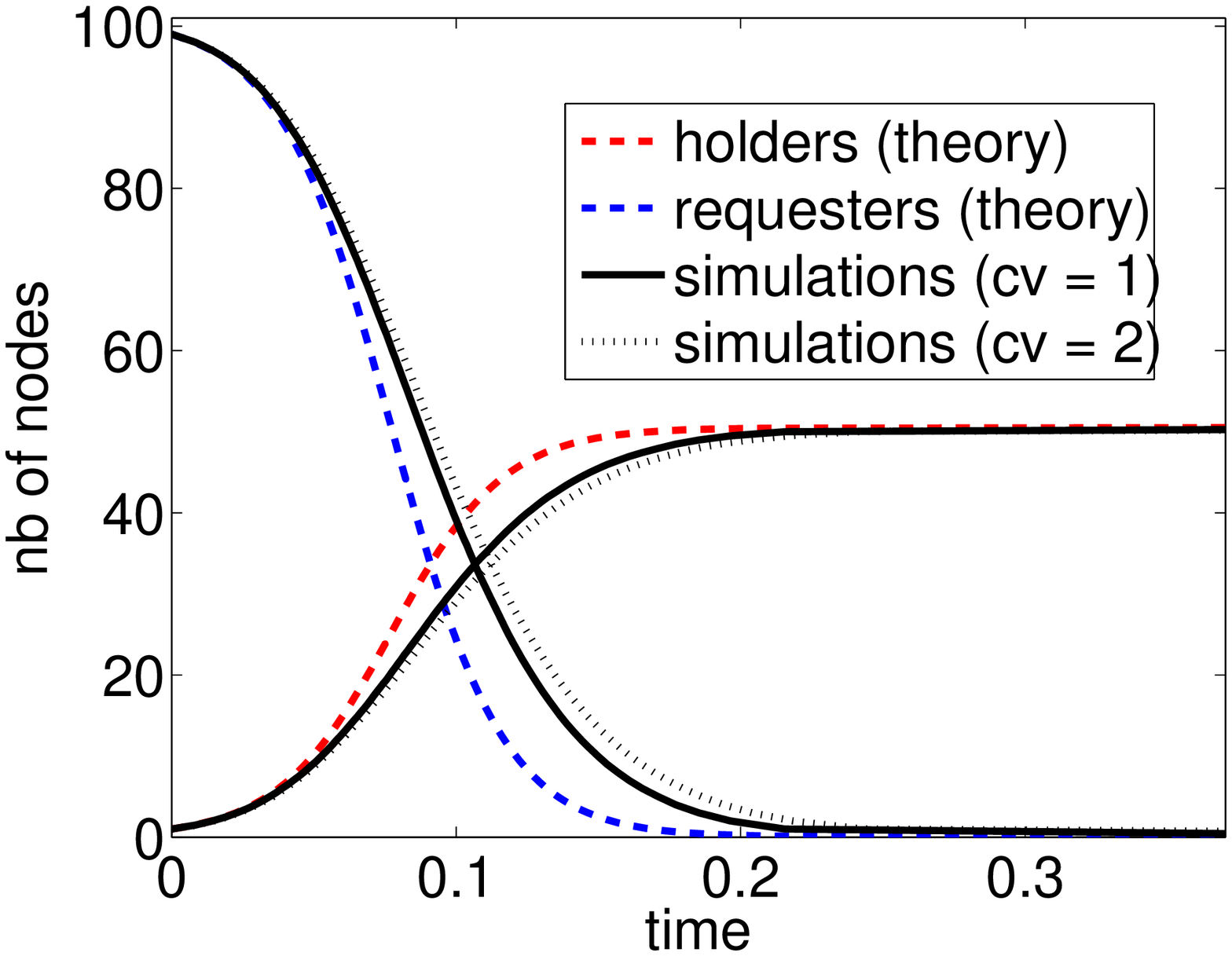}\label{fig:synthetic-H(t)-drop=0}}
\subfigure[$H_{0}=1$, $p_{c}=0.5$]{\includegraphics[width=0.49\linewidth]{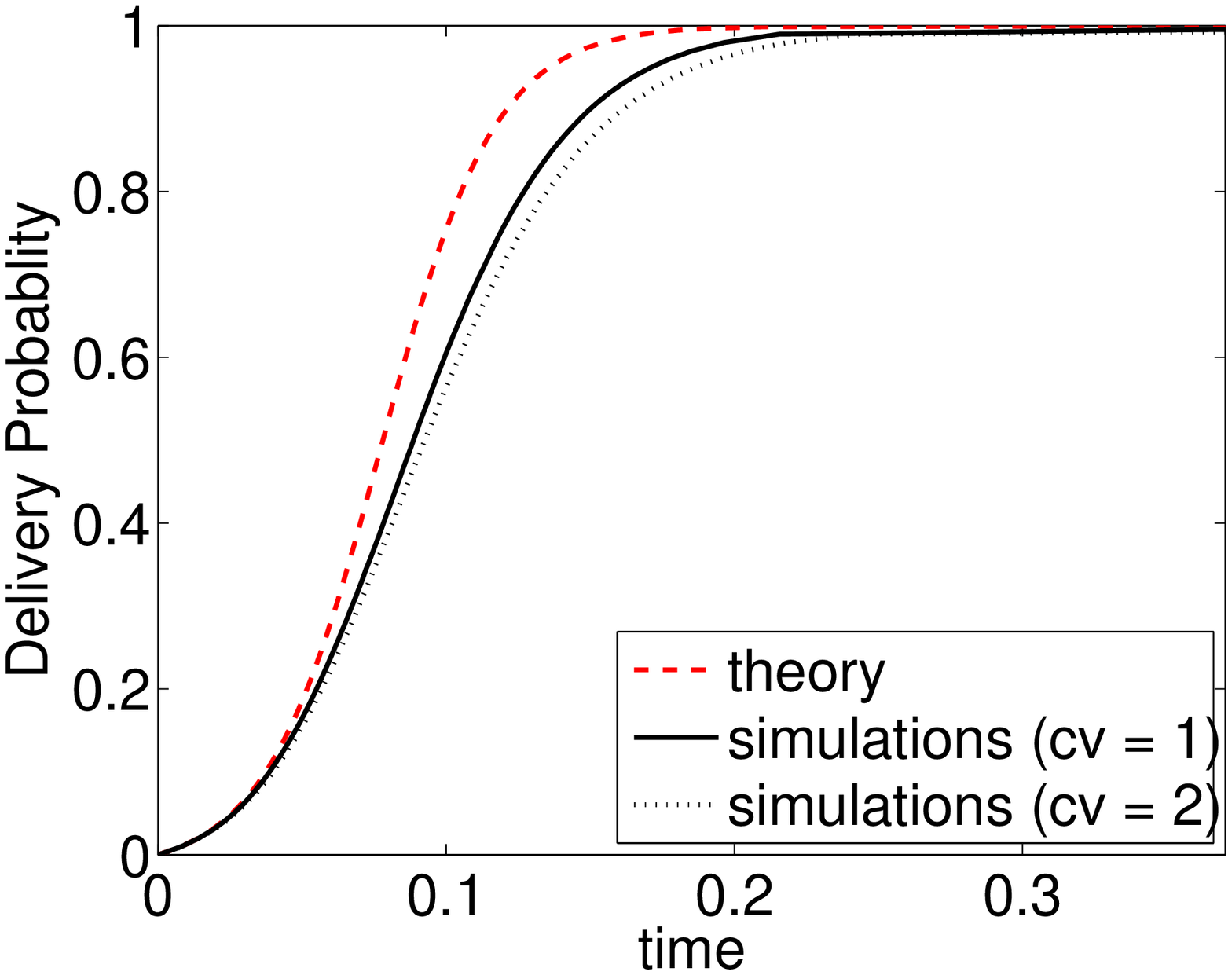}\label{fig:synthetic-Prob-drop=0}}
\caption{(a) Expected number of holders, $H(t)$, and requesters, $R(t)$, over time for generic scenarios with $R_{0}=100$, $H_{SC}=0$; (b) shows the corresponding results for the delivery probability, i.e. $P\{T_{d}\leq TTL\}$, where $TTL$ is the x-axis variable.}
\label{fig:synthetic-H(t)-Prob}
\end{figure}

\textbf{Content Dissemination.} %
In Fig.~\ref{fig:synthetic-H(t)-Prob} we compare simulation results (average values over the different runs) of expected number of holders ($H(t)$) / requesters ($R(t)$) and content delivery probability $P\{T_{d}\leq TTL\}$ with the respective theoretical predictions (Lemma~\ref{thm:ODEs} and Result~\ref{result:delivery-probability-single}, respectively). Considering the same content traffic parameters, we simulated scenarios with moderate ($CV_{\lambda}=1$) and high ($CV_{\lambda}=2$) mobility variance, in order to show how mobility heterogeneity affects the accuracy of our predictions. It can be seen that our predictions become more accurate for lower mobility heterogeneity ($CV_{\lambda}=1$). This is due to the mean field approximation of the transitions rates we used in the analysis (see Section~\ref{sec:single-generic}). For scenarios with even lower mobility heterogeneity (e.g. $CV_{\lambda}=0.5$ - not shown in the plots) the accuracy is even better. Additionally, we need to highlight that these results correspond to an initial allocation of only one holder ($H_{0}=1$), which is the \textit{worst case} scenario (i.e. lowest accuracy of the mean field approximation, and, thus our predictions) among the ones with the given mobility and traffic (other than $H_{0}$) characteristics. In the same scenarios, when considering a few more initial holders, e.g. $H_{0}=10$, theoretical results achieve an almost exact prediction. 

Similar observations can be made in Fig.~\ref{fig:synthetic-Delay}, where we compare the theoretically predicted delivery delays with the respective simulation results. The results in Fig.~\ref{fig:synthetic-Delay} are in accordance with the above observations, i.e. the predictions' accuracy increases for (a) lower $CV_{\lambda}$, and (b) higher number of initial holders $H_{0}$.

\begin{figure}
\subfigure[$p_{c}=0.5$]{\includegraphics[width=0.49\linewidth]{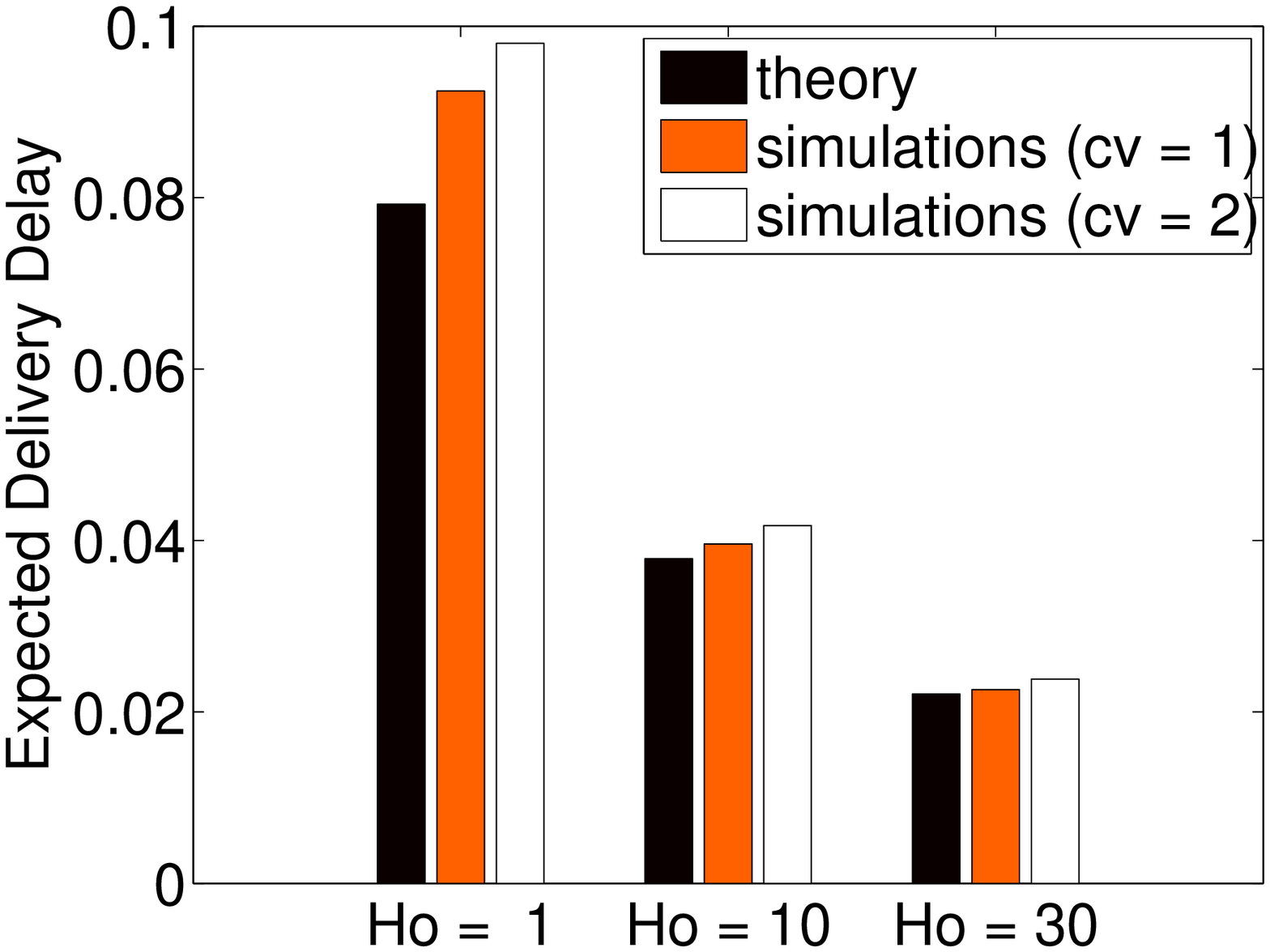}\label{fig:synthetic-Delay-drop=0}}
\subfigure[$p_{c}=1$]{\includegraphics[width=0.49\linewidth]{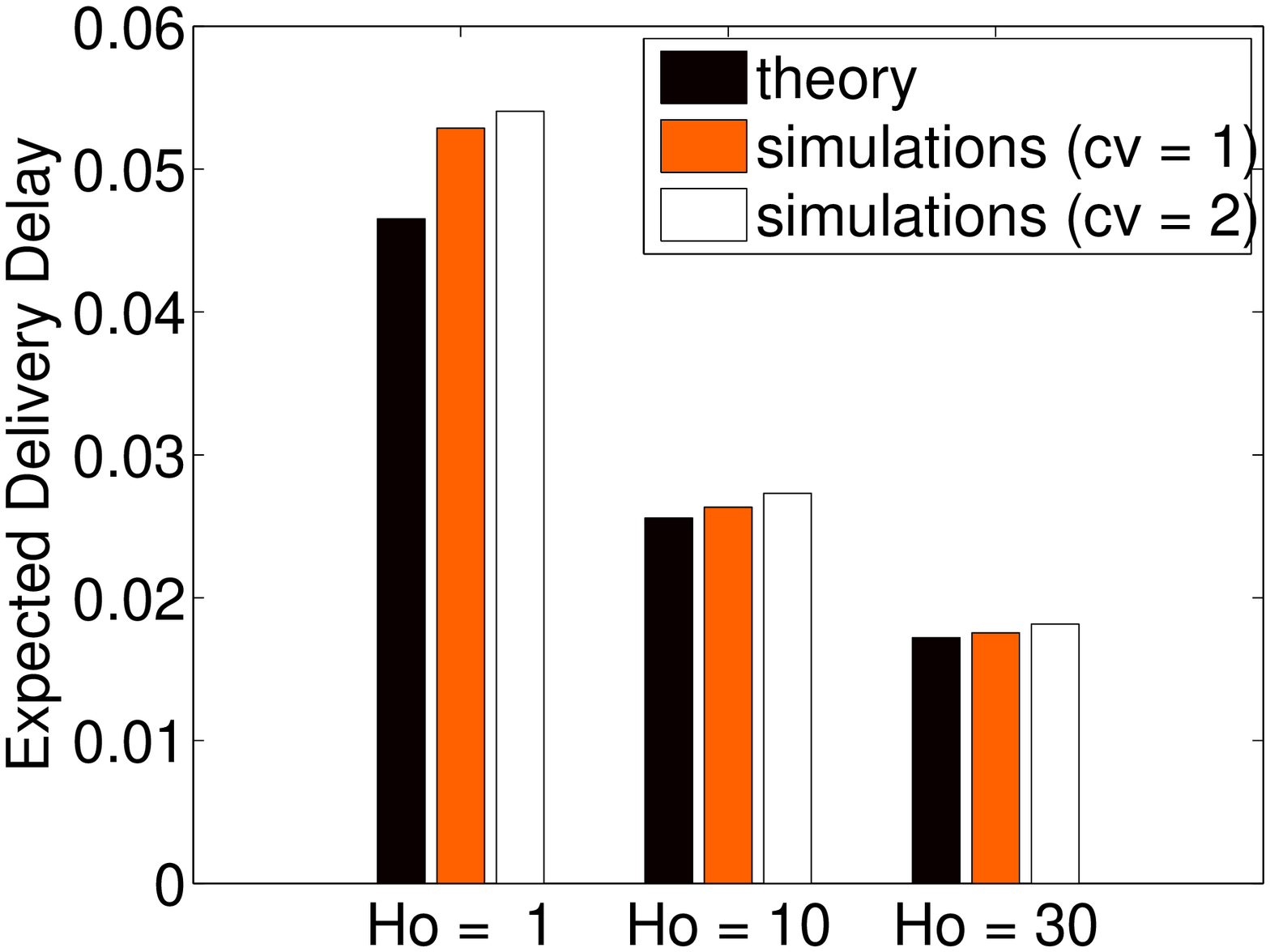}\label{fig:synthetic-Delay-drop>0}}
\caption{Expected delivery delay, $E[T_{d}]$, for various generic scenarios with $R_{0}=100$, $H_{SC}=0$ and (a) $p_{c}=0.5$, (b) $p_{c}=1$.}
\label{fig:synthetic-Delay}
\end{figure}

\textbf{Offloading Cost.} %
We finally present results that validate the cost optimization analysis of Section~\ref{sec:single-cost-optimization}. Fig.~\ref{fig:cost-single-synthetic} shows the incurred cost for the cellular network (y-axis) under different number of initial holders $H_{0}$ (x-axis) for various generic traffic scenarios. Different cooperation policies (top plots: $p_{c}=1$, middle plots: $p_{c}=0.5$, and bottom plots: $p_{c}=0$) and expiry times $TTL$ (or, equivalently, $\gamma = \mu_{\lambda}\cdot TTL$) are considered. It can be seen that our results accurately predict the content dissemination cost.

Some remarkable observations about the optimal initial allocation of holders that can be made in Fig.~\ref{fig:cost-single-synthetic} (as well as in other scenarios we investigated) are the following: (i) In many cases, offloading on the edge can significantly reduce the cost of a content dissemination. For instance, in the scenario shown in Fig.~\ref{fig:cost-single-synthetic} (bottom plot - bottom curve / black color), even without node cooperation ($p_{c}=0$), offloading on the edge can reduce the cost $10$ times, compared to the corresponding scenario without offloading (i.e. $C=100$). (ii) An optimal initial allocation requires only a small number of (initial) storage resources, which in most of the cases we present is equal or less than $20\%$ of the content requesters. (iii) The higher the allowed delay (i.e. expiry time $TTL$ or parameter $\gamma$) is, the larger the gain the cellular network can have is. For example, consider the red line ($\gamma =0.05 $) in the bottom plot. Increasing $\times10$ the value of $TTL$ (black line - $\gamma = 0.5$) can reduce the cost (e.g. for $H_{0}=5$ which is close to the optimal allocation) almost $8$ times.

\begin{figure}
\centering
{\includegraphics[width=0.5\linewidth]{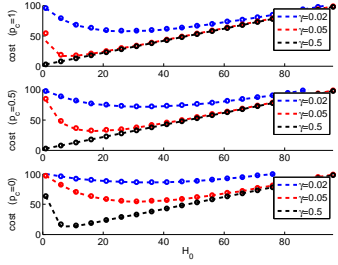}\label{fig:cost-single-synth}}
\caption{Single content offloading cost $C$ (Lemma~\ref{thm:lemma-single-cost}) under different number of initial holders ($H_{0}$, x-axis) for a synthetic mobility scenario with $R_{0}=100$, $H_{SC}=H_{0}$, and $C_{BH}=C_{BS}^{(TTL)}=50 \cdot C_{SC}$. Dashed lines correspond to theoretical predictions and markers to simulation results. We denote $\gamma = \mu_{\lambda}\cdot TTL$.}
\label{fig:cost-single-synthetic}
\end{figure}

\subsubsection{Mobility Traces}

Results of synthetic simulations demonstrate a significant accuracy of our predictions and verify the arguments used in the derivation of our results. In this section, we present results in more challenging scenarios, where node mobility characteristics depart from our model assumptions.

Specifically, we use the TVCM~\cite{tvcm} and SLAW~\cite{slaw} mobility models, which have been shown to capture well real mobility patterns, like power-law flights~\cite{slaw}, community structure~\cite{tvcm}, etc. The generated scenarios we present are\\
\noindent\textit{\textbf{TVCM scenario}}: Mobile nodes move in a square area $1000m\times1000m$, which contains three areas of interest (communities). Nodes move mainly inside their community ($60\%$ of the time) and leave it for a few short periods. Macro-cell BSs provide full coverage of the whole area, while $25$ non-overlapping (placed on a grid) small-cell base stations (SCs), with a communication range of $100m$, provide further connectivity. Mobile nodes are equipped with \textit{D2D} communication interfaces, for which we assume a range of~$30m$.\\
\textit{\textbf{SLAW scenario}}: A square area of edge length $2000m$ is simulated, where mobile nodes either move or remain static for a maximum time of $20min$ (the other mobility parameters are set as in the source code provided by~\cite{slaw}). Macro-cell BSs cover the whole area and coexist with $100$ non-overlapping small-cells. Communication ranges are set as above.


In Fig.~\ref{fig:traces-Prob} we present the delivery probability $P\{T_{d}\leq TTL\}$, along with the theoretical prediction, for two content traffic scenarios in the TVCM (Fig.~\ref{fig:traces-Prob-tvcm}) and SLAW (Fig.~\ref{fig:traces-Prob-slaw}) traces. Contents with popularity $R(0)=50$ are initially cached to $H(0)$ edge nodes (half of which are MNs). The MNs' participation in offloading is set to $p_{c}=0.5$. In the TVCM trace (Fig.~\ref{fig:traces-Prob-tvcm}) it can be seen that the accuracy of our results is significant, despite the community structure of the network (which cannot be captured explicitly by our mobility Assumption~\ref{ass:heterogeneous-mobility}). In the SLAW scenario (Fig.~\ref{fig:traces-Prob-slaw}), our results overestimate the delivery probability. However, note here that the number of holders in the SLAW scenario is smaller, and, thus, our approximation is expected to be less accurate. For scenarios with more initial holders the accuracy of the predictions increase (see e.g. Fig.~\ref{fig:cost-single-SLAW}, where the accuracy is higher for higher $H_{0}$ values). 

\begin{figure}
\subfigure[TVCM: $H(0)=10$, $R(0)=50$]{\includegraphics[width=0.49\linewidth]{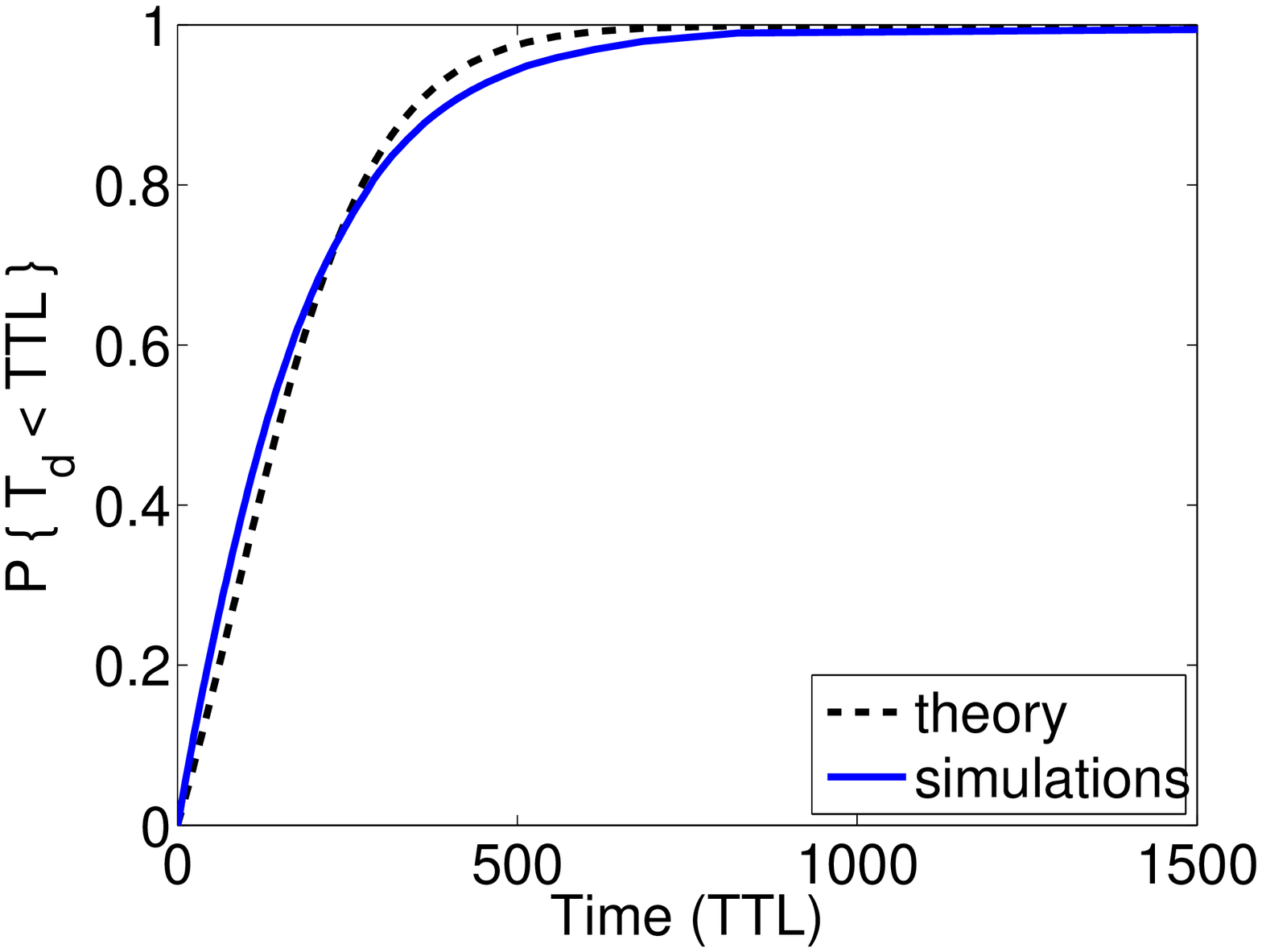}\label{fig:traces-Prob-tvcm}}
\subfigure[SLAW: $H(0)=5$, $R(0)=50$]{\includegraphics[width=0.49\linewidth]{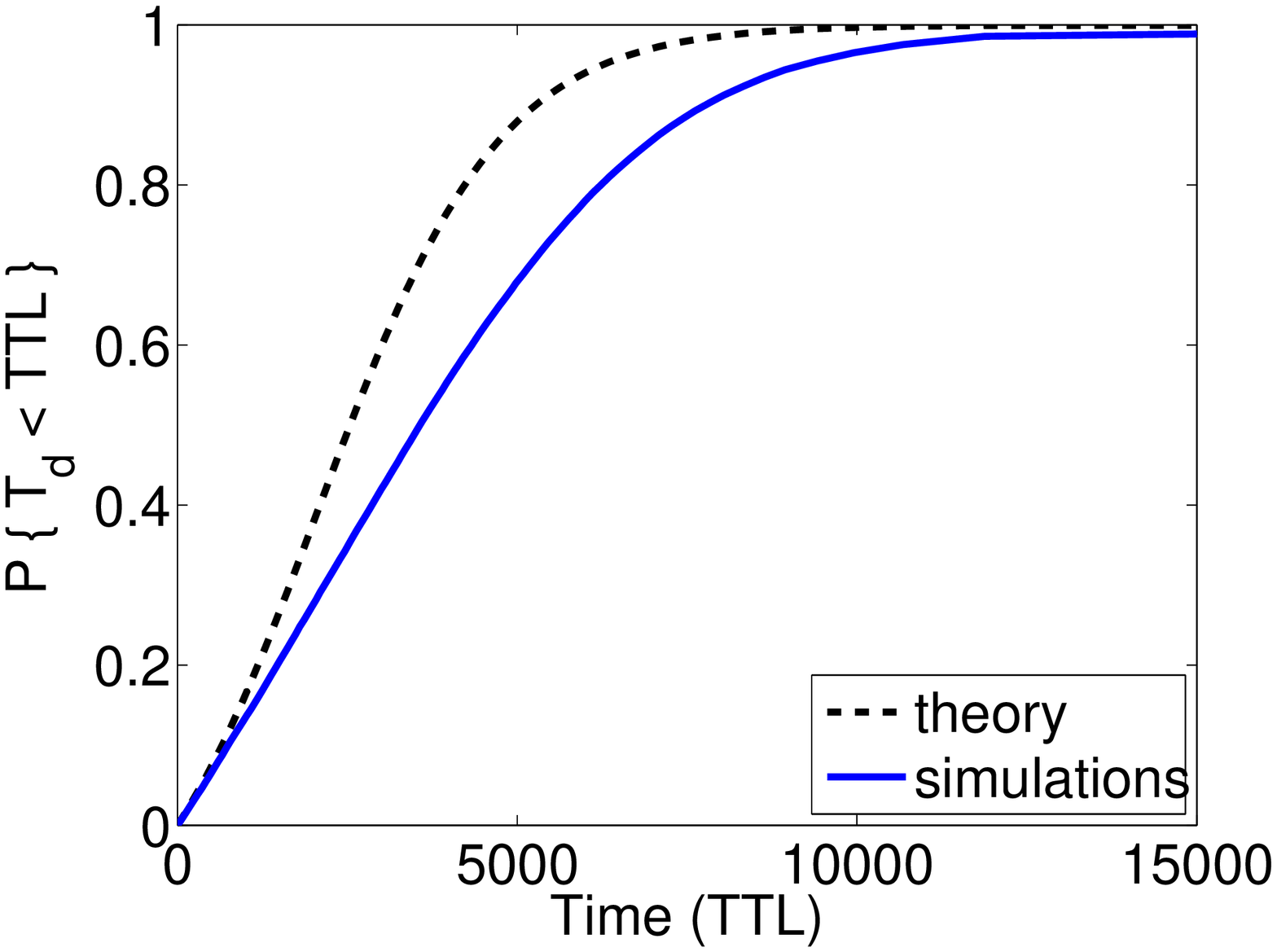}\label{fig:traces-Prob-slaw}}
\caption{Delivery probability $P\{T_{d}\leq TTL\}$ over time $TTL$ (x-axis), for the (a) TVCM and (b) SLAW scenarios with $p_{c}=0.5$ and $H_{SC}(0)=H_{MN}(0)=\frac{H(0)}{2}$.}
\label{fig:traces-Prob}
\end{figure}


Although in some points the theoretical performance metrics deviate considerably from simulations (e.g. $20\%$), the accuracy of the cost metrics (Lemma~\ref{thm:lemma-single-cost}) is less affected. Fig.~\ref{fig:cost-single} shows the incurred cost for delivering a content to $R(0)=30$ requesters (y-axis) under different number of initial holders $H_{0}$ (x-axis). Different initial placement policies ($H_{SC}(0), H_{MN}(0)$), levels of MNs participation ($p_{c}$), and expiry times $TTL$  are considered. In the majority of scenarios our results accurately predict the offloading cost. Yet, even in the case where the predictions are less accurate (e.g. in Fig.~\ref{fig:cost-single-SLAW} for $\mu_{\lambda}\cdot TTL=0.05$), they can still capture the actual optimal initial allocation regimes.
%

\begin{figure}
\subfigure[TVCM ($H_{SC}(0)=H_{MN}(0)$)]{\includegraphics[width=0.49\linewidth]{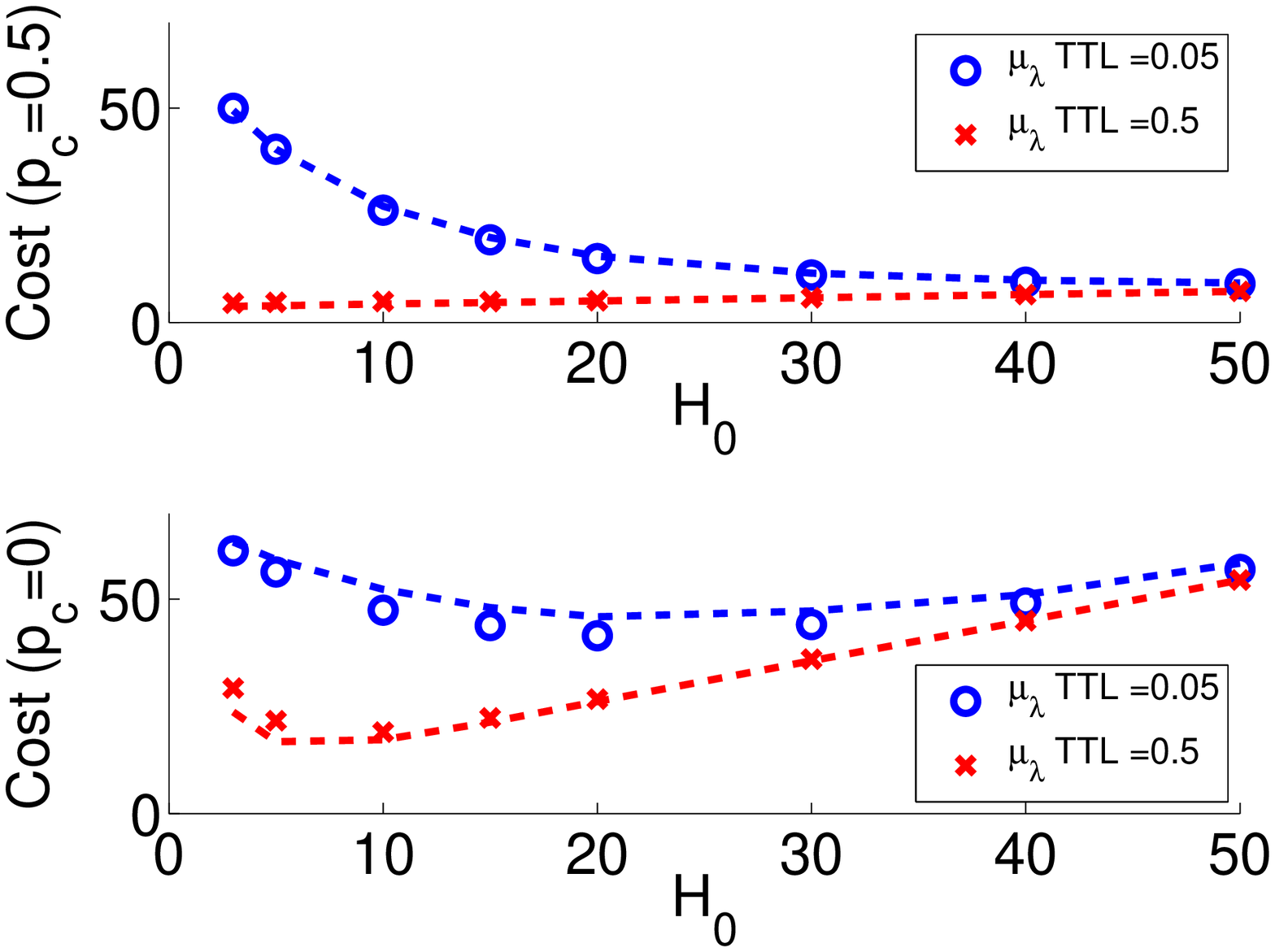}\label{fig:cost-single-TVCM}}
\subfigure[SLAW ($H_{SC}(0)=0$, )]{\includegraphics[width=0.49\linewidth]{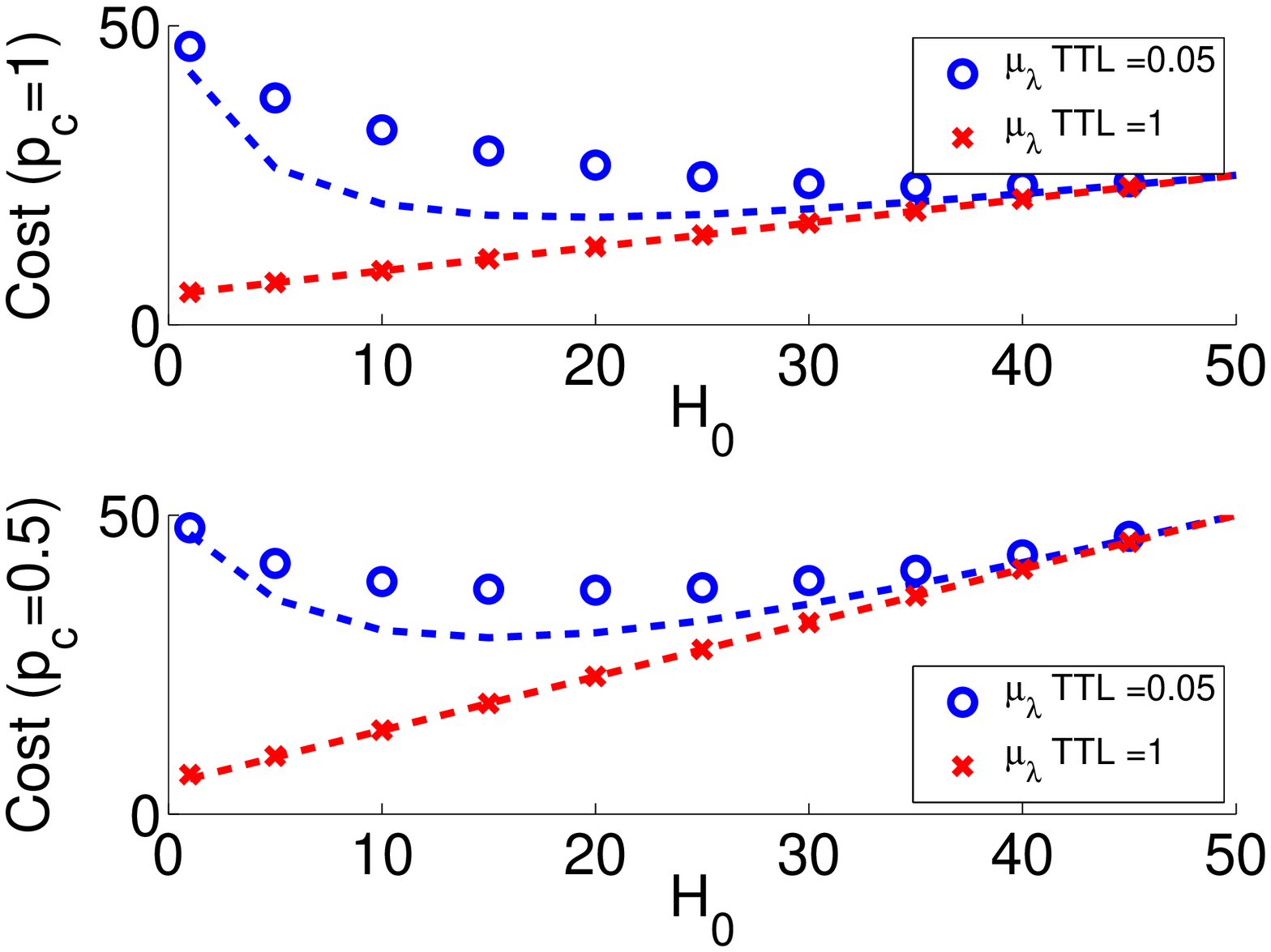}\label{fig:cost-single-SLAW}}
\caption{Offloading cost (y-axis) vs number of initial holders ($H_{0}$, x-axis). Dashed lines correspond to theoretical predictions and markers to simulation results. Transmission costs are: (a) $C_{BS}^{(TTL)}=10\cdot C_{BH}=10\cdot C_{BS}=20\cdot C_{SC}=20\cdot C_{D2D}$ (top plot) and $C_{BS}^{(TTL)}= C_{BH}= C_{BS}=10\cdot C_{SC}$ (bottom plot); (b) $C_{BS}^{(TTL)}=2\cdot C_{BS}=10\cdot C_{D2D}$.}
\label{fig:cost-single}
\end{figure}

\subsection{Performance Evaluation}\label{sec:simulation-cost-efficiency}
After validating our analysis, we now investigate the cost efficiency of the "offloading on the edge" mechanism in a realistic traffic scenario. We present results that demonstrate the effect of different system factors, and provide useful conclusions for cellular network operators.

\noindent The parameters of the scenario we consider are the following:\\
\textit{$-$ Popularity:} Content popularity has been shown to follow a power-law distribution~\cite{youtube-traffic-from-edge,top-video-cellular,pptv-mobile-vod}. Thus, we draw the popularity of each content from a bounded-Pareto distribution ($R_{0}\in[1, 100]$) with shape parameter $\alpha = 0.5$~\cite{youtube-traffic-from-edge}.\\
\textit{$-$ Traffic Intensity:} Mobile operators do not release real mobile traffic data. To this end, and since traffic demand is directly related to the number of mobile users that reside in an area, we infer traffic patterns from an available dataset of the Gowalla \textit{location-based social network}. The Gowalla dataset~\cite{Theus-comcom2012} contains information (logs of position and time) of user checkins (through their mobile devices) in different venues. In the scenarios we present, we create different number of contents during a $24h$ time interval. The number of contents $M$ is proportional to the number of mobile users that checked-in a certain area (we selected the most popular venue) at the same time. The maximum number of concurrent contents is~$M=200$.\\
\textit{$-$ Delay Tolerance:} We set equal expiry times $TTL$ for each content, and we consider different sets of scenarios with low ($TTL=5min$), moderate ($TTL=25min$), and high ($TTL=25min$) delay tolerance.\\
\textit{$-$ Costs:} The relative costs are set $C_{BS} = C_{BS}^{(TTL)} = 10\cdot C_{BH} = 20\cdot C_{SC}= 20\cdot C_{D2D}$, values selected based on some data presented in~\cite{johansson2007cost}.\\
\textit{$-$ Node Mobility:} We use the TVCM mobility scenario presented in the previous section.

\subsection*{\underline{Offloading through SCs}} 
We first consider the case of offloading through SCs. We simulate two sets of scenarios with small ($Q=5$) or large ($Q=200$) caches. We choose the optimal initial caching policy of Result~\ref{THM:OPTIMAL-H0-NO-COOP-MULTIPLE}.

In Fig.~\ref{fig:varying-traffic-demand} we present the total offloading cost (marked lines) incurred for the cellular network operator over different times of the day. The gray area shows the intensity of mobile users that reside in the considered area. The dashed line denotes traffic demand over time, or equivalently, the cost when content delivery \textit{without} offloading is considered.

\noindent Some interesting observations that follow from Fig.~\ref{fig:varying-traffic-demand} are: \\
(i) Under the optimal caching policy, "offloading on the edge" can significantly reduce the cost of content delivery, up to an order of magnitude, or even more in some cases. \\
(ii) The "offloading on the edge" cost changes over time much smoother than traffic demand. In particular, for large caches (cross/red line), the offloading cost curve is almost flat, despite the large peaks in traffic demand. In cellular networks, such temporal variations of the traffic intensity is an important issue, since operators are required to over-provision the network capacity (high CAPEX costs)~\cite{tube}. As we show, "offloading on the edge" can amortize these costs. Even under higher transmission costs $C_{BH}, C_{SC}$ than these we assumed, although the operating cost (OPEX) increases, the cost curve remains smooth, reducing thus a need for over-provisioning. \\
(iii) Large caching capacity has as a result a smoother cost curve (cross/red vs circle/blue curves). This is a positive message for operators, because to equip SCs with large enough caches is both feasible and inexpensive, as discussed in the example scenario of Section~\ref{sec:optimizing-total-cost}. \\
(iv) Comparing Fig.~\ref{fig:varying-traffic-demand-g05} and Fig.~\ref{fig:varying-traffic-demand-g01}, we see that the tolerated delay has also a significant effect on the smoothness of the cost curve (higher $TTL$ values lead to smaller variations). This implies that an alternative way of avoiding the over-provision cost (CAPEX), is to give incentives (OPEX) to users for accepting delayed content. Such solutions have been previously considered, e.g.~\cite{tube}, however, our framework allows an easy investigation of their effects (due to the closed-form results) and an analytic approach of pricing policies, etc.

\begin{figure}
\subfigure[$TTL = 25min$]{\includegraphics[width=0.49\linewidth]{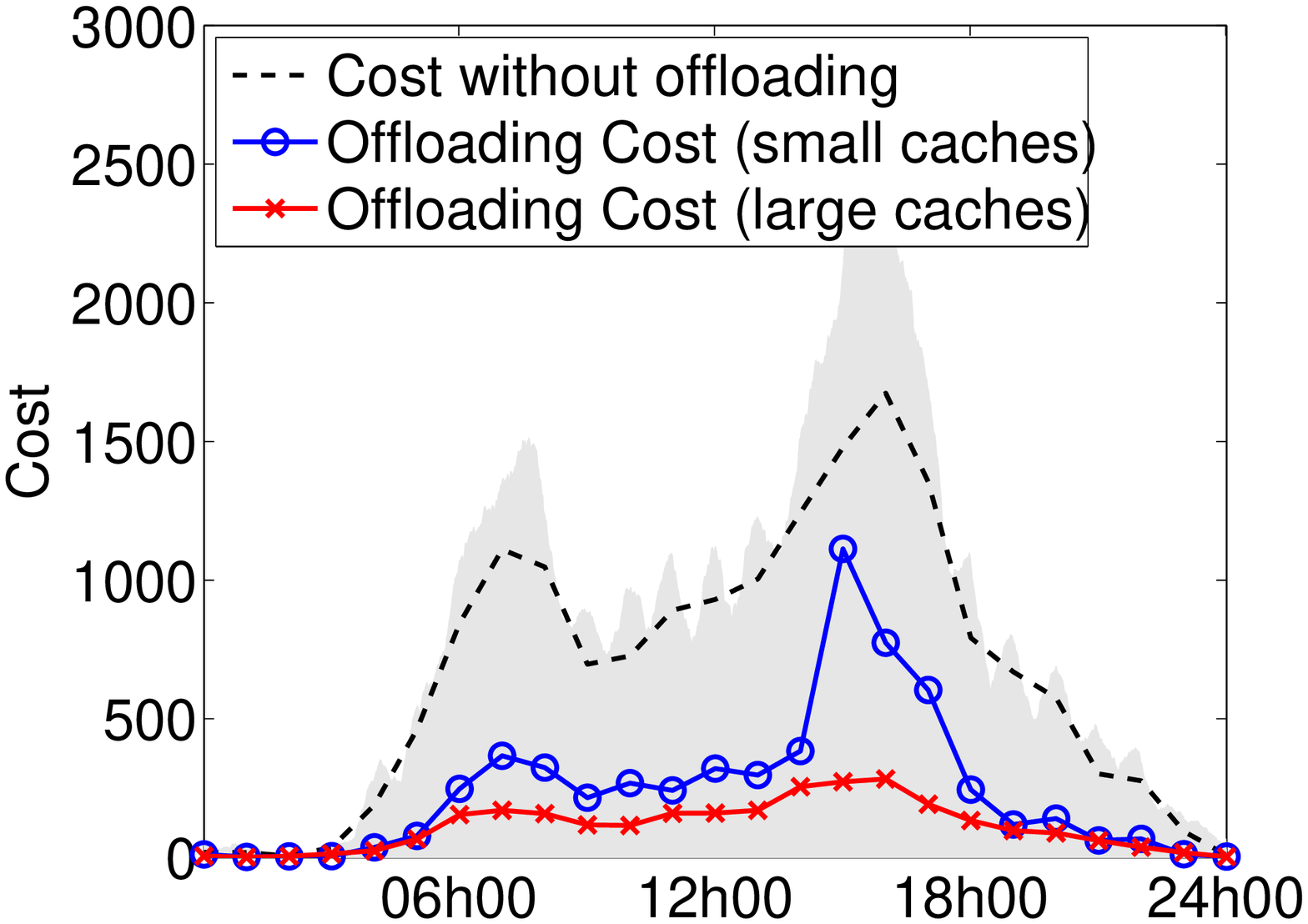}\label{fig:varying-traffic-demand-g05}}
\subfigure[$TTL = 5min$]{\includegraphics[width=0.49\linewidth]{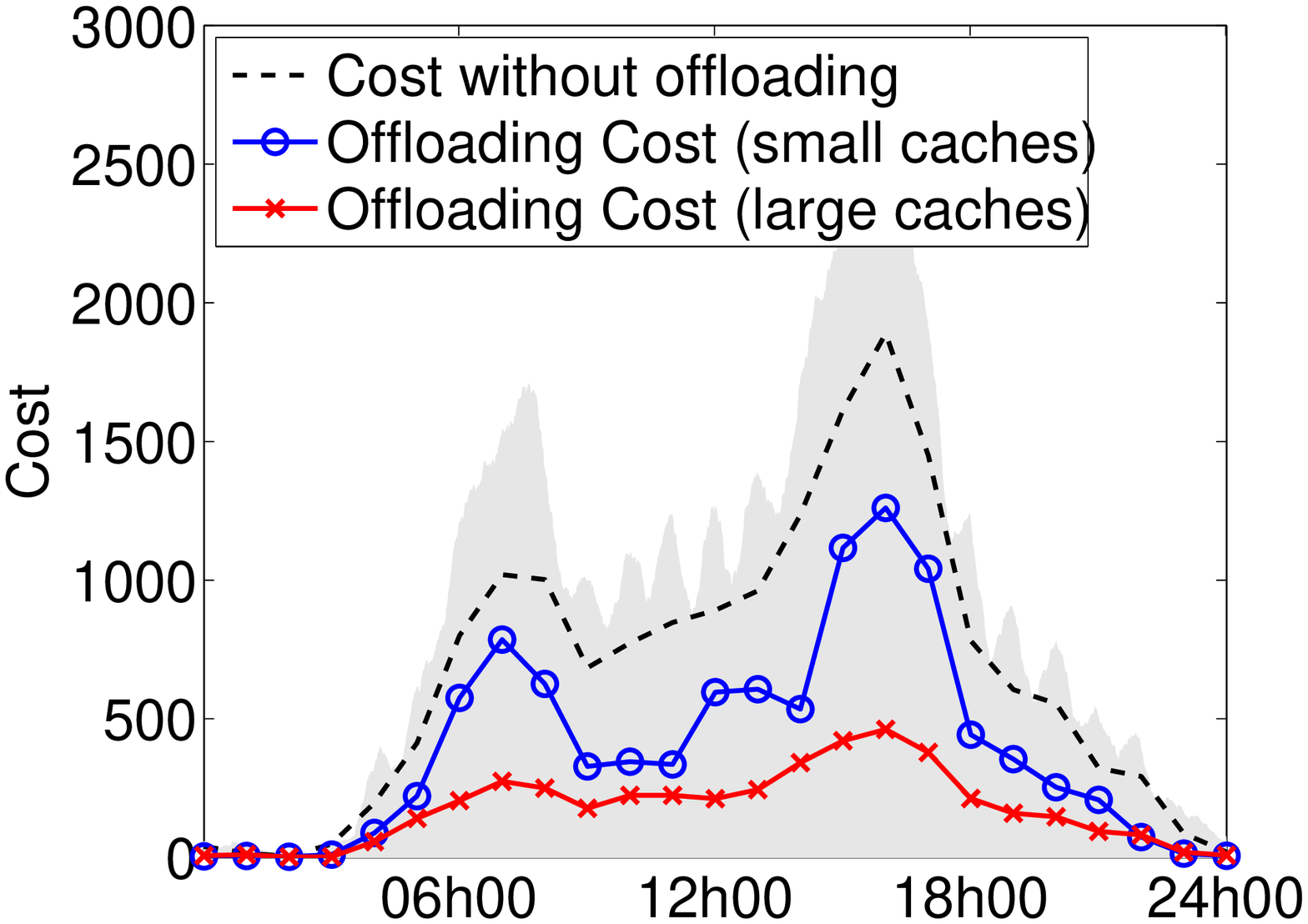}\label{fig:varying-traffic-demand-g01}}
\caption{Traffic demand and offloading cost over a $24h$ period.}
\label{fig:varying-traffic-demand}
\end{figure}

\subsection*{\underline{Offloading through MNs}} 
Now, we evaluate the performance of offloading through MNs. We simulate scenarios with different levels of node cooperation $p_{c}$. We choose the optimal initial content placement policy of Result~\ref{THM:OPTIMAL-H0-MN-MN}.

In Fig.~\ref{fig:varying-traffic-demand-MN-MN-cost-vs-time} we present the total offloading cost (marked lines) incurred for the cellular network operator over different times of the day. We simulate three scenarios with low, moderate and high delay tolerance ($TTL = 5,25,60min$), and $10\%$ of user cooperation in offloading ($p_{c}=0.1$). Similarly to the offloading through SCs case (see e.g. Fig.~\ref{fig:varying-traffic-demand}), for higher $TTL$ values, the cost decreases and its variations are smoother. However, it can be seen that improvement between the scenarios with $TTL=25min$ and $TTL=60min$ is not significant. This has an important implication for the system: Although increasing the delay tolerance is beneficial for the operator, after a point or gradually (depending on the scenario), the effects of this improvement become negligible. Bearing in mind that user satisfaction decreases with $TTL$ indicates that there is a tradeoff, which should be carefully assessed by the system designer or considered for further optimization.

In Fig.~\ref{fig:varying-traffic-demand-MN-MN-cost-vs-time} we show how the total offloading cost over a day period (normalized to the respective cost without offloading) changes with $p_{c}$. It is evident that varying user cooperation does not have the same effects for different scenarios, and that the minimum total cost is achieved at different values of $p_{c}$. This introduces one extra dimension, which can be used for system optimization as well. Such optimization options (with respect to $TTL$, $p_{c}$, etc.) could lead to interesting conclusions, we intend to consider them in future research.

\begin{figure}
\subfigure[Offloading Cost vs. Time]{\includegraphics[width=0.49\linewidth]{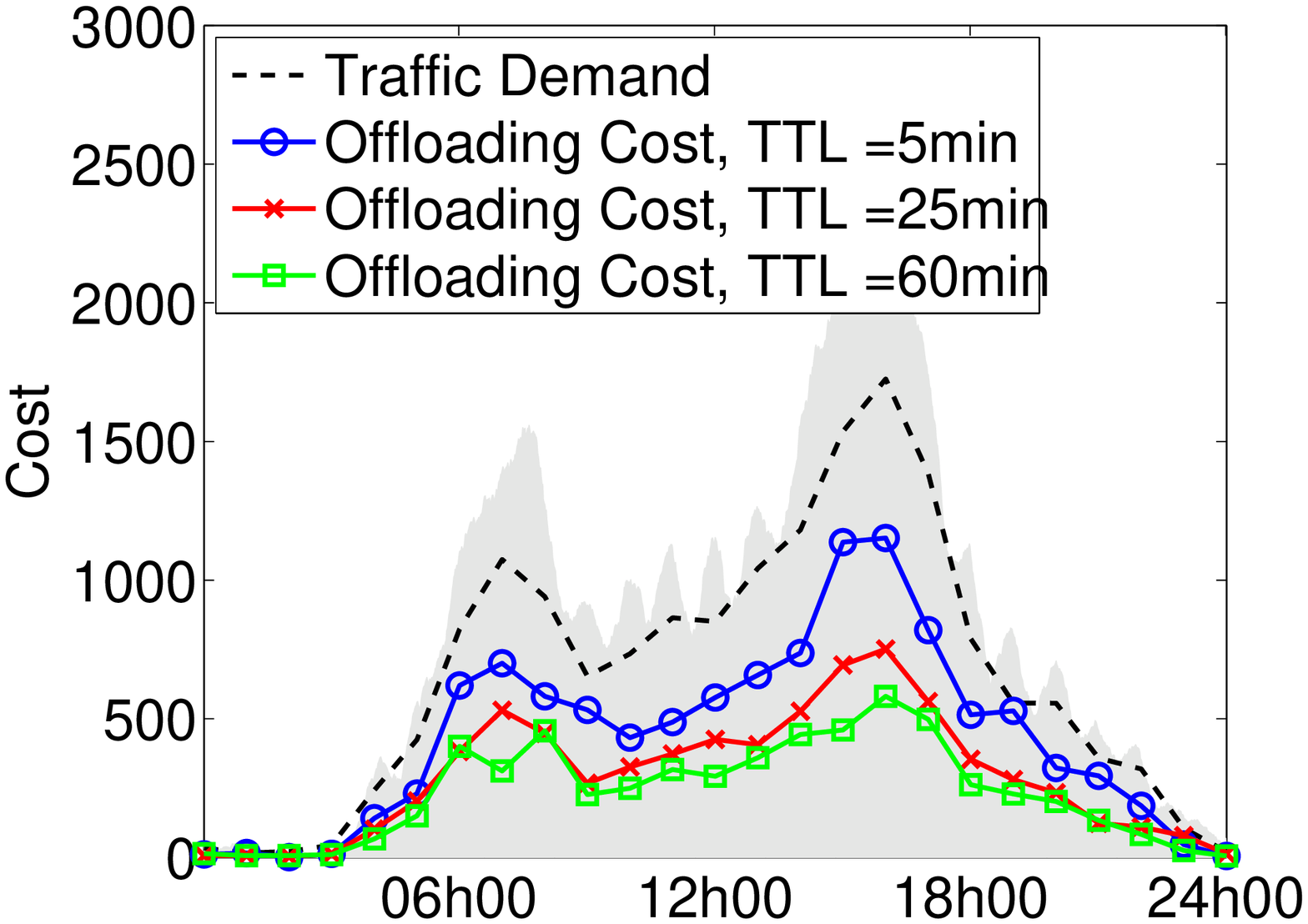}\label{fig:varying-traffic-demand-MN-MN-cost-vs-time}}
\subfigure[Offloading Cost vs. $p_{c}$]{\includegraphics[width=0.49\linewidth]{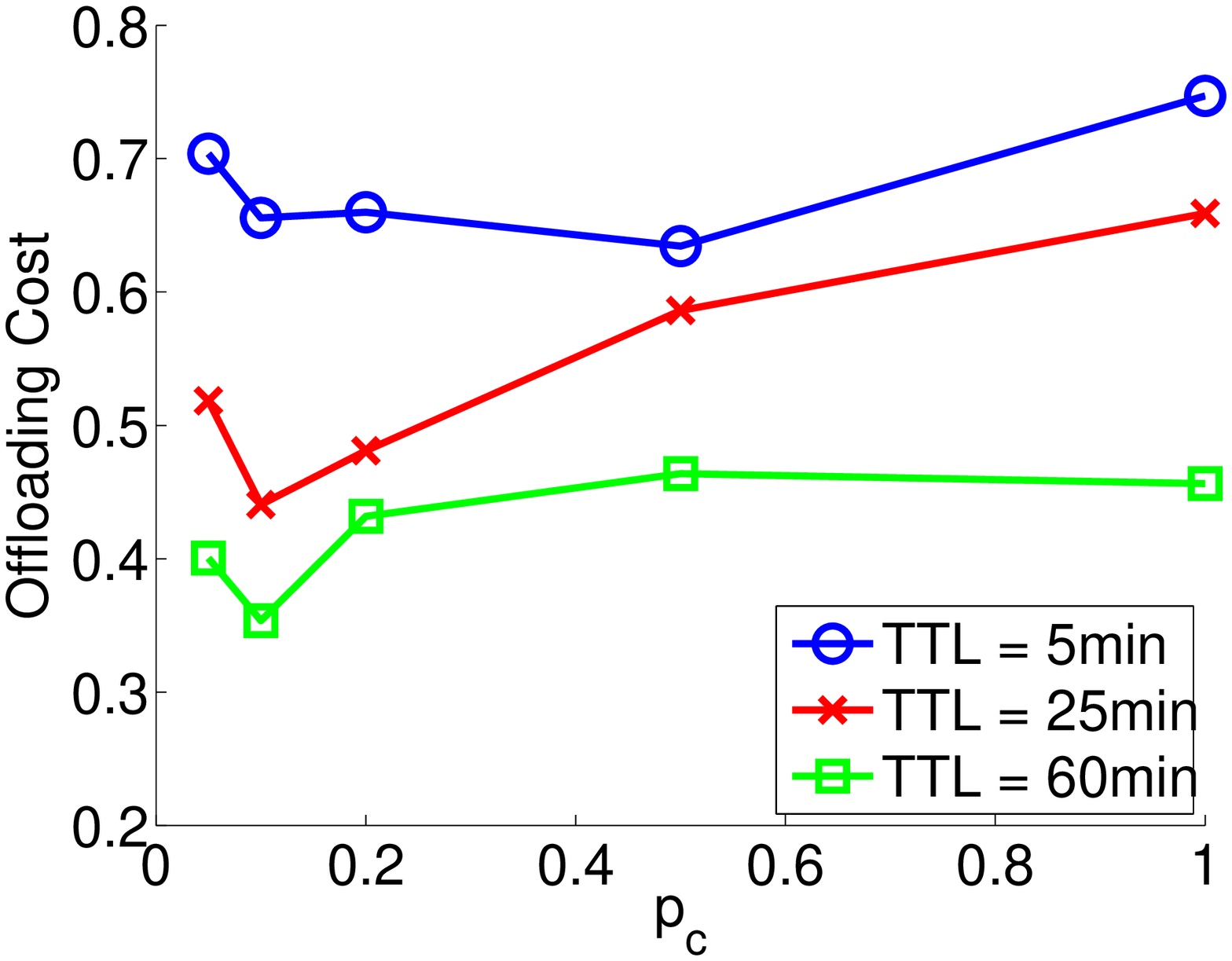}\label{fig:varying-traffic-demand-MN-MN-cost-vs-pc}}
\caption{(a) Traffic demand and offloading cost over a $24h$ period. User cooperation is $10\%$. (b) Total offloading cost over a $24h$ period, normalized to the total cost without offloading.}
\label{fig:varying-traffic-demand-MN-MN}
\end{figure}

\section{Discussion and Extensions}\label{sec:discussion}
In this section we discuss some important issues related to assumptions made implicitly or explicitly throughout our analysis and how they can be extended or removed to make our framework more generic.

\subsection{Heterogeneous Mobility}
The mobility model we use allows heterogeneous meeting rates $\lambda_{ij}$ (Assumption~\ref{ass:heterogeneous-mobility}) in order to account for various node mobility patterns and communication ranges. However, applying the approximation of \eq{eq:mean-value-approximation} (proof of Lemma~\ref{thm:ODEs})\footnote{In \eq{eq:mean-value-approximation} the meeting rates of a sum are approximated with their mean value $\mu_{\lambda}$. As discussed earlier, this approximation becomes more accurate when the heterogeneity of the meeting rates (i.e. the variance of $f_{\lambda}$) decreases.}, leads to considering only the mean value of the meeting rates $\mu_{\lambda}$ in the analytical results. Although the same expressions could have been derived (easier) by using a homogeneous model, i.e. $\forall i,j:\lambda_{ij} = \mu_{\lambda}$, our main motivation for considering heterogeneous rates is the following: We can easily incorporate further heterogeneous social characteristics (related to mobility patterns), like (i) social selfishness (where each node cooperation is related to their social ties)~\cite{pavlos-social-selfishness}, or (ii) smart, mobility-aware content placement algorithms (where "better" holders are selected)~\cite{pavlos-not-all-content}. These characteristics could not have been taken into account under a homogeneous mobility assumption. 

To extend our results for the two above cases, it is just needed to modify the analytic expressions by substituting the average meeting rate $\mu_{\lambda}$ with the effective average meeting rate $\mu_{\lambda}^{(eff.)}$ (see \cite[Lemmas~3.1 and~3.2]{pavlos-social-selfishness} and \cite[Result~4]{pavlos-not-all-content}, respectively), which is given by
\begin{equation}\label{eq:effective-mean-value}
 \mu_{\lambda}^{(eff.)} = E[\lambda\cdot p(\lambda)]~~~\text{and}~~~\mu_{\lambda}^{(eff.)} = \frac{E[\lambda\cdot \pi(\lambda)]}{E[\pi(\lambda)]}
\end{equation}
where the function $p(\lambda)$ describes the social selfishness and $\pi(\lambda)$ the mobility-aware content placement algorithm. The expectations in \eq{eq:effective-mean-value} are taken over the meeting rates distribution $f_{\lambda}$, and, as a result, the mobility heterogeneity is actively involved in the performance prediction expressions.

\subsection{Cost Model}
In our model we considered constant costs (for the different transmission types). However, in different scenarios, there might exist some correlation between transmission costs and other system parameters. Some examples could be: (i) The delayed content delivery cost $C_{BS}^{(TTL)}$ might be a function of $TTL$ (e.g. increasing with $TTL$); (ii) If multicasting is used for initial content placement to MNs, the transmission cost $C_{BS}$ might not be \textit{linearly} related to $H_{MN}(0)$ (i.e. the cost of multicasting a content to $H_{MN}(0)$ nodes, might not be equal to $H_{MN}(0)$ times the cost of a unicast transmission); (iii) The MN-MN transmission cost $C_{D2D}$ might be related to the cooperation probability $p_{c}$, e.g. the willingness of nodes to participate in offloading (which is captured by $p_{c}$) might be higher when the reward for each offloaded content (which is captured by $C_{D2D}$) increases.

Our results predicting the content dissemination performance and cost (Section~\ref{sec:analysis-single-content}) hold also when adopting such more generic cost models. What changes is the cost optimization problem (Problem~\ref{eq:optimization-problem-multi}), which might need to be reformulated. However, even under the above examples (as well as a number of other cost models), Problem~\ref{eq:optimization-problem-multi} can be still expressed in a closed form expression and thus be solved with known (analytic or numerical) methods. 

Summarizing, since till now there is no common technology or cost policies applied in offloading systems, we prefered to assume a simple cost model with constant costs. Nevertheless, as discussed above, our results can be extended for other cases as well.

\subsection{Content Dissemination}
The assumptions for the content dissemination (e.g. non-increasing number of requesters, only initial placement to SCs, etc.) can be extended and, thus, more generic scenarios can be captured and analysed in a similar manner. In the remainder, we demonstrate how our model and analysis is extended to include two extra characteristics, namely (i) content discards and (ii) bulk arrivals/departures of requesters.

First, we extend to cases where contents can be discarded by edge nodes before their expiry time ($TTL$). For example, a MN whose battery level decreases, might stop storing and forwarding contents in order to save energy. As a second example, when a MN's cache is full and a new content is received, the MN needs to drop either one of the stored contents or the one just received. On the other hand, the state of the SCs' caches will be probably known by the cellular network and, thus, such cases can be avoided or controlled. To be able to analyze such content discards with a Markovian framework, we use the following model
\begin{assumption}[Content Dropping]\label{ass:dropping} A MN drops a cached content with rate $\lambda_{d}\geq0$. The content dropping process is Poisson and $\lambda_{d}$ is equal among all holders and all contents.
\end{assumption}
Under the above assumption, the ODE of \eq{eq:ODE-h(t)} can be turned into
\begin{equation}\label{eq:ODE-h(t)-generic}
\frac{dH(t)}{dt} = p_{c}\cdot H(t)\cdot R(t)\cdot \mu_{\lambda} - (H(t)-H_{SC}(0))\cdot \lambda_{d}
\end{equation}
where the last term corresponds to the content discards by the MN-holders.

Second, we assume that at times $\{\tau_{1}, \tau_{2}, ...\}$ either a number of new MNs get interested in a content (i.e. new requesters enter the system) or some of the existing requesters lose their interest in the content (i.e. they leave the system). We denote the numbers of such arriving/departing requesters at times $\{\tau_{1}, \tau_{2}, ...\}$ as $\{R_{\tau_{1}}, R_{\tau_{2}}, ...\}$, where a value $R_{\tau}$ can be positive (denoting arrivals) or negative (denoting departures). As a result, the ODE for the number of requesters $R(t)$ over time (see \eq{eq:ODE-r(t)}) is given now by
\begin{equation}\label{eq:ODE-r(t)-generic}
\frac{dR(t)}{dt} = -H(t)\cdot R(t)\cdot \mu_{\lambda} + \sum_{\tau}R_{\tau}\cdot\delta(t-\tau)
\end{equation}
where $\delta(\cdot)$ is the \textit{Dirac delta function}.

Finally, solving the system of the ODEs of \eq{eq:ODE-h(t)-generic} and \eq{eq:ODE-r(t)-generic}, gives the solution for the deterministic approximations for $H(t)$ and $R(t)$ for a generic scenario with content discards and bulk arrivals/departures of requesters.

\section{Related Work}\label{sec:related}
In this section we discuss works that are closer to ours, rather than studies which do not consider caching and/or delay tolerant delivery, and which are mainly based on pure infrastructure architectures, e.g. with WiFi access points~\cite{Offloading-Wifi} or small-cell base stations~\cite{femtocell-survey,HetNets-paradigm}, or on the D2D paradigm~\cite{survey-d2d}.

Mobile data offloading through opportunistic communications and epidemic content dissemination is studied in~\cite{offloading-wowmom11,offloading-control-theory,fluid-limit-mass2012,offloading-double-opportunities}. In the setting of~\cite{offloading-wowmom11}, copies of a content are distributed through the infrastructure to a subset of mobile nodes, which then start propagating them epidemically. The performance of different content ``pushing'' techniques (e.g. slow/fast start) is investigated through simulations on a real vehicular mobility trace. Analytical approaches for pushing techniques can be found in~\cite{offloading-control-theory,fluid-limit-mass2012}, which study the optimal selection of the number of initial and final content pushes. \cite{offloading-control-theory} models the content dissemination as a control system and proposes an adaptive algorithm, \textit{HYPE}, which aims to minimize the load of the cellular network by using real time measurements. On the other hand,~\cite{fluid-limit-mass2012} uses a fluid limit approximation and focuses on the cost optimization problem. Finally,~\cite{offloading-double-opportunities} takes into account fairness among different contents/nodes, and derives schedulers that maximize the throughput, under given mobility and wireless channel conditions. These studies, in contrast to our framework, assume that \textit{every} user is willing to offload contents, even if they are \textit{not of her interest}. Difficulties in devising incentive mechanisms or limitations of device capabilities, might render such settings unrealistic.

To this end,~\cite{multiple-offloading,pavlos-not-all-content} consider a limited number of (designated) holders. \cite{multiple-offloading} proposes centralized algorithms for selecting the best set of available holders, in order to minimize the traffic load served by the infrastructure. In a different approach,~\cite{pavlos-not-all-content} focuses on the effects of \textit{popularity} (number of requesters) and \textit{availability} (number of holders) on the performance of content delivery. Our paper extends these works, by introducing generic offloading costs and policies, and deriving insightful, closed-form results for the optimal~caching.

Finally,~\cite{femtocaching-magazine} proposes caching in femto-cells and user devices, in a different setting than ours, where users communicate with several holders simultaneously. D2D communication is controlled by a macro-cell BS, which is aware of the status of caches, location of users, and channel state information between them. The objective of the paper is to decide which files should be stored and on which helper node, a problem that is shown to be \textit{NP-hard}. This problem is formally presented, studied in more detail, and extended for coded contents in~\cite{femtocaching}.

\section{Conclusion}
In this work we studied ``offloading on the edge'', a mechanism that employs edge nodes (SCs and/or MNs) to opportunistically offload popular content. We built a model that can capture heterogeneous traffic demand, user cooperation and mobility characteristics, and describe generic caching and offloading policies. Based on our model, we derived closed-form expressions for predicting the offloading performance. These allowed us to analytically study the cost optimization problem, and provide results that shed light on how caching policies should be designed. Realistic simulations verified~the insights that stem from our analysis, and led to useful conclusions.

Our closed-form expressions reveal how and to what extent each system parameter affects performance and cost. Thus, they could be easily applied to sensitivity analysis, network planning and dimensioning, or design of pricing strategies; issues that have recently attracted a lot of attention from network operators, who seek novel solutions to alleviate the effects of the rapidly growing traffic demand.

\bibliographystyle{IEEEtran}

\appendix
\subsection{Proof of Result~\ref{result:expected-delay-single}}\label{appendix:proof-delay}
\begin{proof}
The probability a content to be delivered in the time interval $[t,t+dt)$ is given by
\begin{equation}\label{eq:P-Td-equal-t}
P\{T_{d}\in[t,t+dt)\}= \frac{dP\{T_{d}\leq t\}}{dt}\cdot dt 
\end{equation}

Since a requester gets the content at time $t=TTL$ from a BS, if it has not received it earlier, we can write for the expected delay
\begin{multline}\label{eq:expectation-conditional-generic}
E[T_{i}|TTL] =TTL\cdot (1-P\{T_{d}\leq TTL\}) \\+ \int_{0}^{TTL}t\cdot P\{T_{d}\in[t,t+dt)\}\\
			 = TTL\cdot (1-P\{T_{d}\leq TTL\}) + \int_{0}^{TTL}t\cdot \frac{dP\{T_{d}\leq t\}}{dt} \cdot dt
\end{multline}
where the last equality follows from \eq{eq:P-Td-equal-t}.

Using the expression of Result~\ref{result:delivery-probability-single}, we first compute the derivative $\frac{dP\{T_{d}\leq t\}}{dt}$, and, then, the integral in \eq{eq:expectation-conditional-generic}, and we get
\begin{multline*}
E[T_{i}|TTL] 
= 	TTL\cdot (1-P\{T_{d}\leq TTL\}) \\
	+ \frac{1}{p_{c}\cdot R_{0}}\cdot 
\left(\frac{TTL\cdot H_{0}\cdot (p_{c}\cdot R_{0}+H_{0})\cdot e^{\mu_{\lambda}\cdot(p_{c}\cdot R_{0}+H_{0})\cdot TTL }}{p_{c}\cdot R_{0}+H_{0}\cdot e^{\mu_{\lambda}\cdot(p_{c}\cdot R_{0}+H_{0})\cdot TTL }} \right) 	\\
	+ \frac{1}{\mu_{\lambda}\cdot p_{c}\cdot R_{0}}\cdot \ln\left(\frac{p_{c}\cdot R_{0}+H_{0}}{p_{c}\cdot R_{0}+H_{0}\cdot e^{\mu_{\lambda}\cdot(p_{c}\cdot R_{0}+H_{0})\cdot TTL }}\right)
\end{multline*}
Substituting the value of $P\{T_{d}\leq TTL\}$ from Result~\ref{result:delivery-probability-single} in the above equation, after some algebraic manipulations, we can successively get
\begin{multline*}
E[T_{i}|TTL] 
= 	\frac{TTL\cdot (p_{c}\cdot R_{0}+H_{0})}{p_{c}\cdot R_{0}}  	\\
	+ \frac{1}{\mu_{\lambda}\cdot p_{c}\cdot R_{0}}\cdot \ln\left(\frac{p_{c}\cdot R_{0}+H_{0}}{p_{c}\cdot R_{0}+H_{0}\cdot e^{\mu_{\lambda}\cdot(p_{c}\cdot R_{0}+H_{0})\cdot TTL }}\right)\\
	=\frac{1}{\mu_{\lambda}\cdot p_{c}\cdot R_{0}}\cdot \ln\left(\frac{(p_{c}\cdot R_{0}+H_{0})\cdot e^{\mu_{\lambda}\cdot(p_{c}\cdot R_{0}+H_{0})\cdot TTL }}{p_{c}\cdot R_{0}+H_{0}\cdot e^{\mu_{\lambda}\cdot(p_{c}\cdot R_{0}+H_{0})\cdot TTL }}\right)\\
	=\frac{1}{\mu_{\lambda}\cdot p_{c}\cdot R_{0}}\cdot \ln\left(1+\frac{p_{c}\cdot R_{0}- e^{-\mu_{\lambda}\cdot(p_{c}\cdot R_{0}+H_{0})\cdot TTL }}{H_{0}+p_{c}\cdot R_{0}\cdot e^{-\mu_{\lambda}\cdot(p_{c}\cdot R_{0}+H_{0})\cdot TTL }}\right)
\end{multline*}
which is the expression of Result~\ref{result:expected-delay-single} for $p_{c}>0$. The expression for $p_{c}=0$ follows after taking the limit ($p_{c}\rightarrow0$) of the above expression.
\end{proof}

\subsection{Lemma~\ref{thm:monotonicity}: Cost Monotonicity with $\lambda_{0}$}\label{appendix:lemma-monotonicity}
\begin{lemma}\label{thm:monotonicity}
Under a content placement policy given by \eq{eq:optimal-Ho}, the derivative of the total cost, $\sum_{\theta\in\mathcal{M}}C^{\theta}$, with respect to $\lambda_{0}$ is
\begin{equation*}\label{eq:derivative-cost-final}
\frac{d}{d\lambda_{0}}\left[ \sum_{\theta\in\mathcal{M}}C^{\theta}\right] = \frac{1}{\gamma}\cdot\left(1-\frac{1}{1+\frac{\lambda_{0}}{\Phi_{1}}}\right)\cdot |A|	\geq 0
\end{equation*}
where $\mathcal{A} = \lbrace\theta\in\mathcal{M}:L\leq R_{0}^{\theta}\leq U\rbrace$.
\end{lemma}
\begin{proof}
From the conditions (b) and (c) (see, proof of Result~\ref{THM:OPTIMAL-H0-NO-COOP-MULTIPLE}), and similarly to \eqs{eqs:L-U}, we can express the multipliers $\lambda_{\theta}$ and $\mu_{\theta}$ as a function of $\lambda_{0}$, as
\begin{subequations}\label{eqs:langragian-parameters}
 \begin{align}
 \lambda_{\theta} &= \left\{
 \begin{array}{lc}
  \lambda_{0}+C_{BH}\left(1-\gamma\cdot \Phi\cdot R_{0}^{\theta}\right)	&, R_{0}^{\theta}<L\\
  0		&, R_{0}^{\theta}\geq L
 \end{array}
\right. \\
 \mu_{\theta} &= \left\{
 \begin{array}{lc}
  -\lambda_{0}-C_{BH}\left(1-\gamma\cdot\Phi\cdot e^{-\gamma\cdot N_{SC}} R_{0}^{\theta}\right)	& , R_{0}^{\theta}>U\\
  0		& , R_{0}^{\theta}\leq U
 \end{array}
\right. 
\end{align}
\end{subequations}

The cost of a single content dissemination, \eq{eq:cost-single-for-optimization}, under the content placement policy of \eq{eq:optimal-Ho}, can be written as
\begin{align*}
C^{\theta} &= \frac{\Phi_{1}}{\gamma}\cdot\left[\ln\left(\gamma\cdot\Phi\cdot R_{0}^{\theta}\right)-\ln\left(1+\frac{\lambda_{0}-\lambda_{\theta}+\mu_{\theta}}{\Phi_{1}} \right) \right ] \\
	&+ \Phi_{2}\cdot R_{0}^{\theta}\\
	&+(\Phi_{3}-\Phi_{2})\cdot R_{0}^{\theta}\cdot\frac{1}{\gamma\cdot\Phi\cdot R_{0}^{\theta}} \cdot \left(1+\frac{\lambda_{0}-\lambda_{\theta}+\mu_{\theta}}{\Phi_{1}} \right) \\
&= \frac{\Phi_{1}}{\gamma}\cdot\left[\ln\left(\gamma\cdot\Phi\cdot R_{0}^{\theta}\right)-\ln\left(1+\frac{\lambda_{0}-\lambda_{\theta}+\mu_{\theta}}{\Phi_{1}} \right) \right ] \\
	&+ \Phi_{2}\cdot R_{0}^{\theta}+\frac{\Phi_{1}}{\gamma}\cdot \left(1+\frac{\lambda_{0}-\lambda_{\theta}+\mu_{\theta}}{\Phi_{1}} \right)\numberthis
\end{align*}
Taking its derivative, with respect to $\lambda_{0}$, gives
\begin{align*}\label{eq:derivative-cost-lambda}
 \frac{d}{d\lambda_{0}}\left[ \sum_{\theta\in\mathcal{M}}C^{\theta}\right] 
	&=-\frac{\Phi_{1}}{\gamma}\cdot\frac{d}{d\lambda_{0}}\left[\sum_{\theta\in\mathcal{M}}\ln\left(1+\frac{\lambda_{0}-\lambda_{\theta}+\mu_{\theta}}{\Phi_{1}} \right ) \right ] \\
	&+ \frac{1}{\gamma}\cdot\frac{d}{d\lambda_{0}} \left[\sum_{\theta\in\mathcal{M}}\left( \lambda_{0}-\lambda_{\theta}+\mu_{\theta}\right )\right]\numberthis
\end{align*}
because the terms including only the scenario parameters ($R_{0}^{\theta}$, $\gamma$, and costs) do not depend on the selected resource allocation and, thus, on the parameter $\lambda_{0}$.

To calculate the derivatives appearing in the right side of \eq{eq:derivative-cost-lambda}, we use the definition of a derivative, i.e.
\begin{equation}
 \frac{df(\lambda_{0})}{d\lambda_{0}} = \lim_{d\lambda_{0}\rightarrow 0}\frac{f(\lambda_{0}+d\lambda_{0})-f(\lambda_{0})}{d\lambda_{0}}
\end{equation}
and proceed as following:

We first define the sets
\begin{subequations}
 \begin{align}
 \mathcal{A} &= \lbrace\theta\in\mathcal{M}:L\leq R_{0}^{\theta}\leq U\rbrace\\
\mathcal{B} &= \lbrace\theta\in\mathcal{M}: R_{0}^{\theta}> U\rbrace\\
\mathcal{C} &= \lbrace\theta\in\mathcal{M}: R_{0}^{\theta}< L\rbrace
\end{align}
\end{subequations}
and, respectively, for $\lambda_{0}\rightarrow \lambda_{0}+d\lambda_{0}$, the sets
\begin{subequations}
 \begin{align}
 \mathcal{A}^{'} &= \lbrace\theta\in\mathcal{M}:L+\Delta L \leq R_{0}^{\theta}\leq +\Delta U\rbrace\\
\mathcal{B}^{'} &= \lbrace\theta\in\mathcal{M}: R_{0}^{\theta}> U+\Delta U \rbrace\\
\mathcal{C}^{'} &= \lbrace\theta\in\mathcal{M}: R_{0}^{\theta}< L+\Delta L \rbrace
\end{align}
\end{subequations}
where we denoted
\begin{subequations}\label{eqs:delta-L-U-definitions}
 \begin{align}
L+ \Delta L &= \frac{1}{\gamma\cdot\Phi}\cdot \left(1+\frac{\lambda_{0}+d\lambda_{0}}{C_{BH}}\right) = L + \frac{d\lambda_{0}}{\gamma\cdot C_{BH}\cdot \Phi}\\
U+ \Delta U&= \frac{1}{\gamma\cdot\Phi}\cdot e^{\gamma\cdot N_{SC}}\cdot \left(1+\frac{\lambda_{0}+d\lambda_{0}}{C_{BH}}\right) \nonumber\\
	   = U + &\frac{d\lambda_{0}}{\gamma\cdot C_{BH}\cdot \Phi}\cdot e^{\gamma\cdot N_{SC}} = (L+\Delta L)\cdot e^{\gamma\cdot N_{SC}}
\end{align}
\end{subequations}

Regarding the first derivative term in \eq{eq:derivative-cost-lambda}, we proceed as following
\begin{align*}\label{eq:derivative-sum-1}
\frac{d}{d\lambda_{0}}&\left[\sum_{\theta\in\mathcal{M}}\ln\left(1+\frac{\lambda_{0}-\lambda_{\theta}+\mu_{\theta}}{C_{BH}} \right ) \right ] \\
\stackrel{\text{\eqs{eqs:langragian-parameters}}}{=} 
	&\frac{d}{d\lambda_{0}}\left[\sum_{\theta\in\mathcal{A}}\ln\left(1+\frac{\lambda_{0}}{C_{BH}} \right) \right]\\ 
	&+ \frac{d}{d\lambda_{0}}\left[\sum_{\theta\in\mathcal{B}}\left(\ln\left(\gamma\cdot\Phi\cdot R_{0}^{\theta} \right)-\gamma\cdot N_{SC}\right)\right]\\
	&+ \frac{d}{d\lambda_{0}}\left[\sum_{\theta\in\mathcal{C}}\ln\left(\gamma\cdot\Phi\cdot R_{0}^{\theta} \right) \right]\\
=
&\frac{d}{d\lambda_{0}}\left[|\mathcal{A}|\ln\left(1+\frac{\lambda_{0}}{C_{BH}}\right) \right]\\
&+ \frac{d}{d\lambda_{0}}\left[\sum_{\theta\in\mathcal{B}}\ln\left(\gamma\cdot\Phi\cdot R_{0}^{\theta} \right) \right]-\gamma\cdot N_{SC}\cdot \frac{d|\mathcal{B}|}{d\lambda_{0}}\\
&+ \frac{d}{d\lambda_{0}}\left[\sum_{\theta\in\mathcal{C}}\ln\left(\gamma\cdot\Phi\cdot R_{0}^{\theta} \right) \right]\\
=
&|\mathcal{A}|\cdot \frac{1}{C_{BH}}\cdot \frac{1}{1+\frac{\lambda_{0}}{C_{BH}}}+\ln\left(1+\frac{\lambda_{0}}{C_{BH}}\right)\cdot \frac{d|\mathcal{A}|}{d\lambda_{0}}\\
&+ \frac{d}{d\lambda_{0}}\left[\sum_{\theta\in\mathcal{B}}\ln\left(\gamma\cdot\Phi\cdot R_{0}^{\theta} \right) \right]-\gamma\cdot N_{SC}\cdot \frac{d|\mathcal{B}|}{d\lambda_{0}}\\
&+ \frac{d}{d\lambda_{0}}\left[\sum_{\theta\in\mathcal{C}}\ln\left(\gamma\cdot\Phi\cdot R_{0}^{\theta} \right) \right] \numberthis
\end{align*}

The derivatives in the above sum are calculated as following
\begin{subequations}\label{eqs:derivatives-for-sum-1}
\begin{align*}
 \frac{d|\mathcal{A}|}{d\lambda_{0}} &= \frac{|\mathcal{A}^{'}|-|\mathcal{A}|}{d\lambda_{0}}\\
	&= \frac{\int_{L+\Delta L}^{U+\Delta U}M\cdot \rho(x)dx-\int_{L}^{U}M\cdot \rho(x)dx}{d\lambda_{0}}\\
	&= M\cdot \frac{\int_{U}^{U+\Delta U}\rho(x)dx-\int_{L}^{L+\Delta L}\rho(x)dx}{d\lambda_{0}}\\
	&\approx M\cdot \frac{p(U)\cdot \Delta U-p(L)\cdot \Delta L}{d\lambda_{0}}\\
	&\stackrel{\text{\eqs{eqs:delta-L-U-definitions}}}{=} M\cdot \frac{p(U)\cdot \Delta L \cdot e^{\gamma\cdot N_{SC}}-p(L)\cdot \Delta L}{d\lambda_{0}}\\
	&= M\cdot \frac{\Delta L}{d\lambda_{0}}\cdot \left(p(U)\cdot e^{\gamma\cdot N_{SC}}-p(L)\right)\\
	&\stackrel{\text{\eqs{eqs:delta-L-U-definitions}}}{=} \frac{M}{\gamma\cdot  C_{BH}\cdot\Phi}\cdot \left(p(U)\cdot e^{\gamma\cdot N_{SC}}-p(L)\right)\numberthis
\end{align*}
and, similarly, 
\begin{equation}
\frac{d|\mathcal{B}|}{d\lambda_{0}} \approx - M\cdot \frac{e^{\gamma\cdot N_{SC}}}{\gamma\cdot  C_{BH}\cdot\Phi}\cdot p(U)
\end{equation}
and 
\begin{align*}
& \frac{d}{d\lambda_{0}}\left[\sum_{\theta\in\mathcal{B}}\ln\left(\gamma\cdot\Phi\cdot R_{0}^{\theta} \right) \right] \\
	&= \frac{\sum_{\theta\in\mathcal{B}^{'}}\ln\left(\gamma\cdot\Phi\cdot R_{0}^{\theta} \right)-\sum_{\theta\in\mathcal{B}}\ln\left(\gamma\cdot\Phi\cdot R_{0}^{\theta} \right)}{d\lambda_{0}}\\
	&= \frac{-\int_{U}^{U+\Delta U}\ln(\gamma\cdot\Phi\cdot x)\cdot M\cdot \rho(x)dx}{d\lambda_{0}}\\
	&\approx -M\cdot \frac{\ln(\gamma\cdot\Phi\cdot U)\cdot p(U)\cdot \Delta U}{d\lambda_{0}}\\
	&\stackrel{\text{\eqs{eqs:delta-L-U-definitions}}}{=} - M\cdot \frac{e^{\gamma\cdot N_{SC}}}{\gamma\cdot  C_{BH}\cdot\Phi}\cdot \ln(\gamma\cdot\Phi\cdot U)\cdot p(U)\\
	&\stackrel{\text{\eqs{eqs:L-U}}}{=} - M\cdot \frac{e^{\gamma\cdot N_{SC}}}{\gamma\cdot  C_{BH}\cdot\Phi}\cdot p(U) \cdot \left(\gamma\cdot N_{SC}+\left(1+\frac{\lambda_{0}}{C_{BH}}\right)\right)\numberthis
\end{align*}
and, similarly, 
\begin{align*}
 \frac{d}{d\lambda_{0}}&\left[\sum_{\theta\in\mathcal{C}}\ln\left(\gamma\cdot\Phi\cdot R_{0}^{\theta} \right) \right] \\
	&\approx M\cdot \frac{1}{\gamma\cdot  C_{BH}\cdot\Phi}\cdot p(L)\cdot \ln\left(1+\frac{\lambda_{0}}{C_{BH}}\right)\numberthis
\end{align*}
\end{subequations}

Substituting \eqs{eqs:derivatives-for-sum-1} in \eq{eq:derivative-sum-1}, gives 
\begin{equation}\label{eq:derivative-sum-1-final}
 \frac{d}{d\lambda_{0}}\left[\sum_{\theta\in\mathcal{M}}\ln\left(1+\frac{\lambda_{0}-\lambda_{\theta}+\mu_{\theta}}{\Phi_{1}} \right ) \right ] = |\mathcal{A}| \cdot \frac{1}{\Phi_{1}}\cdot \frac{1}{1+\frac{\lambda_{0}}{\Phi_{1}}}
\end{equation}

Regarding the second derivative term in \eq{eq:derivative-cost-lambda}, we proceed as following
\begin{align*}\label{eq:derivative-sum-2}
\frac{d}{d\lambda_{0}}&\left[\sum_{\theta\in\mathcal{M}}\left(\lambda_{0}-\lambda_{\theta}+\mu_{\theta}\right) \right ] \\
\stackrel{\text{\eqs{eqs:langragian-parameters}}}{=} 
	&\frac{d}{d\lambda_{0}}\left[\sum_{\theta\in\mathcal{A}}\lambda_{0}\right]+\frac{d}{d\lambda_{0}}\left[\sum_{\theta\in\mathcal{B}}(\lambda_{0}+\mu_{\theta})\right]+\frac{d}{d\lambda_{0}}\left[\sum_{\theta\in\mathcal{C}}(\lambda_{0}-\lambda_{\theta})\right]\\ 
=
	&\frac{d}{d\lambda_{0}}\left[\lambda_{0}\cdot |\mathcal{A}|\right]\\
	&+\frac{d}{d\lambda_{0}}\left[\sum_{\theta\in\mathcal{B}}\left(-C_{BH}+\gamma\cdot C_{BH}\cdot\Phi\cdot e^{-\gamma\cdot N_{SC}} \cdot R_{0}^{\theta}\right)\right]\\
	&+\frac{d}{d\lambda_{0}}\left[\sum_{\theta\in\mathcal{C}}\left(-C_{BH}+\gamma\cdot C_{BH}\cdot\Phi\cdot R_{0}^{\theta}\right)\right]\\
=
	&|\mathcal{A}|+ \lambda_{0}\cdot \frac{d|\mathcal{A}|}{d\lambda_{0}}\\
	&-C_{BH}\cdot \frac{d|\mathcal{B}|}{d\lambda_{0}}+\gamma\cdot C_{BH}\cdot\Phi\cdot e^{-\gamma\cdot N_{SC}} \cdot \frac{d}{d\lambda_{0}}\left[\sum_{\theta\in\mathcal{B}}R_{0}^{\theta}\right]\\
	&-C_{BH}\cdot \frac{d|\mathcal{C}|}{d\lambda_{0}}+\gamma\cdot C_{BH}\cdot\Phi\cdot \frac{d}{d\lambda_{0}}\left[\sum_{\theta\in\mathcal{C}}R_{0}^{\theta}\right]\numberthis
\end{align*}
Similarly as before, we get
\begin{subequations}\label{eqs:derivatives-for-sum-2}
\begin{equation}
 \frac{d|\mathcal{C}|}{d\lambda_{0}} \approx M\cdot \frac{1}{\gamma\cdot  C_{BH}\cdot\Phi}\cdot p(L)
\end{equation}
and 
\begin{align*}
 \frac{d}{d\lambda_{0}}&\left[\sum_{\theta\in\mathcal{B}}R_{0}^{\theta}\right] = \frac{-\int_{U}^{U+\Delta U}x\cdot M\cdot \rho(x)dx}{d\lambda_{0}}\\
	&\approx -M\cdot \frac{U\cdot p(U)\cdot \Delta U}{d\lambda_{0}}\\
	&\stackrel{\text{\eqs{eqs:delta-L-U-definitions}}}{=} -M\cdot \frac{\Delta L}{d\lambda_{0}}\cdot L\cdot p(U)\cdot e^{2\cdot \gamma\cdot N_{SC}}\\
	&\stackrel{\text{\eqs{eqs:L-U}}}{=} -M\cdot \frac{e^{\gamma\cdot N_{SC}}}{\gamma\cdot  C_{BH}\cdot\Phi}\cdot \frac{1}{\gamma\cdot\Phi} \cdot \left(1+\frac{\lambda_{0}}{C_{BH}}\right)p(U)\numberthis
\end{align*}
and, similarly, 
\begin{align*}
 \frac{d}{d\lambda_{0}}\left[\sum_{\theta\in\mathcal{C}}R_{0}^{\theta}\right] \approx -M\cdot \frac{1}{\gamma\cdot  C_{BH}\cdot\Phi}\cdot \frac{1}{\gamma\cdot\Phi} \cdot \left(1+\frac{\lambda_{0}}{C_{BH}}\right)p(L)\numberthis
\end{align*}
\end{subequations}

Substituting \eqs{eqs:derivatives-for-sum-2} in \eq{eq:derivative-sum-2}, gives 
\begin{equation}\label{eq:derivative-sum-2-final}
 \frac{d}{d\lambda_{0}} \left[\sum_{\theta\in\mathcal{M}}\left( \lambda_{0}-\lambda_{\theta}+\mu_{\theta}\right )\right] = |\mathcal{A}|
\end{equation}

Finally, substituting the expressions of \eq{eq:derivative-sum-1-final} and \eq{eq:derivative-sum-2-final} in \eq{eq:derivative-cost-lambda}, proves the Lemma.
\end{proof}

\end{document}